\documentclass[sigconf,authorversion]{acmart}

\usepackage{algorithm2e}  %
\usepackage{amsthm}  %
\usepackage{multirow}  %

\AtBeginDocument{%
  \providecommand\BibTeX{{%
    \normalfont B\kern-0.5em{\scshape i\kern-0.25em b}\kern-0.8em\TeX}}}

\setcopyright{othergov}
\copyrightyear{2020}
\acmYear{2020}
\acmDOI{10.1145/3427228.3427297}

\acmConference[ACSAC 2020]{Annual Computer Security
Applications Conference}{December 7--11, 2020}{Austin, USA}
\acmBooktitle{Annual Computer Security Applications Conference (ACSAC 2020), December 7--11, 2020, Austin, USA}
\acmPrice{15.00}
\acmISBN{978-1-4503-8858-0/20/12}

\begin{document}

\title{FPSelect: Low-Cost Browser Fingerprints for Mitigating Dictionary Attacks against Web Authentication Mechanisms}

\author{Nampoina Andriamilanto}
\orcid{0000-0002-0224-5664}
\email{tompoariniaina.andriamilanto@irisa.fr}
\affiliation{%
  \institution{Univ Rennes, CNRS, IRISA}
  \streetaddress{263 avenue du général Leclerc}
  \city{Rennes}
  \country{France}
  \postcode{35000}
}
\additionalaffiliation{%
  \institution{IRT b$<>$com}
  \streetaddress{1219 avenue Champs Blancs}
  \city{Cesson-Sévigné}
  \country{France}
  \postcode{35510}
}

\author{Tristan Allard}
\orcid{0000-0002-2777-0027}
\email{tristan.allard@irisa.fr}
\affiliation{%
  \institution{Univ Rennes, CNRS, IRISA}
  \streetaddress{263 avenue du général Leclerc}
  \city{Rennes}
  \country{France}
  \postcode{35000}
}

\author{Ga\"etan Le Guelvouit}
\email{gaetan.leguelvouit@b-com.com}
\affiliation{%
  \institution{IRT b$<>$com}
  \streetaddress{1219 avenue Champs Blancs}
  \city{Cesson-Sévigné}
  \country{France}
  \postcode{35510}
}

\renewcommand{\shortauthors}{Andriamilanto, et al.}

\begin{abstract}
  Browser fingerprinting consists into collecting attributes from a web browser.
  Hundreds of attributes have been discovered through the years.
  Each one of them provides a way to distinguish browsers, but also comes with a usability cost (e.g., additional collection time).
  In this work, we propose FPSelect, an attribute selection framework allowing verifiers to tune their browser fingerprinting probes for web authentication.
  We formalize the problem as searching for the attribute set that satisfies a security requirement and minimizes the usability cost.
  The security is measured as the proportion of impersonated users given a fingerprinting probe, a user population, and an attacker that knows the exact fingerprint distribution among the user population.
  The usability is quantified by the collection time of browser fingerprints, their size, and their instability.
  We compare our framework with common baselines, based on a real-life fingerprint dataset, and find out that in our experimental settings, our framework selects attribute sets of lower usability cost.
  Compared to the baselines, the attribute sets found by FPSelect generate fingerprints that are up to $97$ times smaller, are collected up to $3,361$ times faster, and with up to $7.2$ times less changing attributes between two observations, on average.
\end{abstract}



\maketitle

\newenvironment{proofsketch}{  %
  \renewcommand{\proofname}{Proof Sketch}\proof}{\endproof}

\section{Introduction}
  Nowadays, the managers of web platforms face a crucial choice about which authentication mechanism to use.
  On the one hand, using solely passwords is common and easy, but fallible due to the many attacks that exist: brute force, dictionary~\cite{WAMG09}, credential stuffing~\cite{TLZBRIMCEMMPB17}, or targeted knowledge attacks~\cite{WZWYH16}.
  Previous studies report over $3.3$~billion credentials leaked~\cite{MIL18}, a password reuse rate above $30$\%~\cite{HLNGX18}, and the risk of account hijacking increased by $400$~times if the credentials of an account are stolen through a phishing attack~\cite{TLZBRIMCEMMPB17}.
  On the other hand, using supplementary authentication factors improves security, but at the cost of usability~\cite{BHOS12}.
  Indeed, users are required to remember, possess, or undergo additional actions, which is impractical in real-life (e.g, using a security token requires users to constantly carry it).
  As a result, few users authenticate using multiple factors (e.g., it is estimated that less than $10$\% of the active Google accounts use two factors~\cite{PTAI15, MIL18}).

  Initially used to track users on the web, \emph{browser fingerprinting} has recently been identified as a promising authentication factor~\cite{UMFHSW13, PJ15, SPJ15, AV16, SPJ17, AAL21}.
  It consists into collecting the values of attributes from a web browser (e.g., the \texttt{UserAgent} HTTP header~\cite{RFC7231UserAgent}, the screen resolution, the way it draws a picture~\cite{MS12}) to build a fingerprint.
  This technique is already endorsed by open source access management solutions (e.g., OpenAM\footnote{
    \url{https://backstage.forgerock.com/docs/am/6.5/authentication-guide/\#device-id-match-hints}
  }) and by software products (e.g., SecureAuth\footnote{
    \url{https://docs.secureauth.com/x/agpjAg}
  }).

  \begin{figure}
    \centering
    \includegraphics[width=\columnwidth]{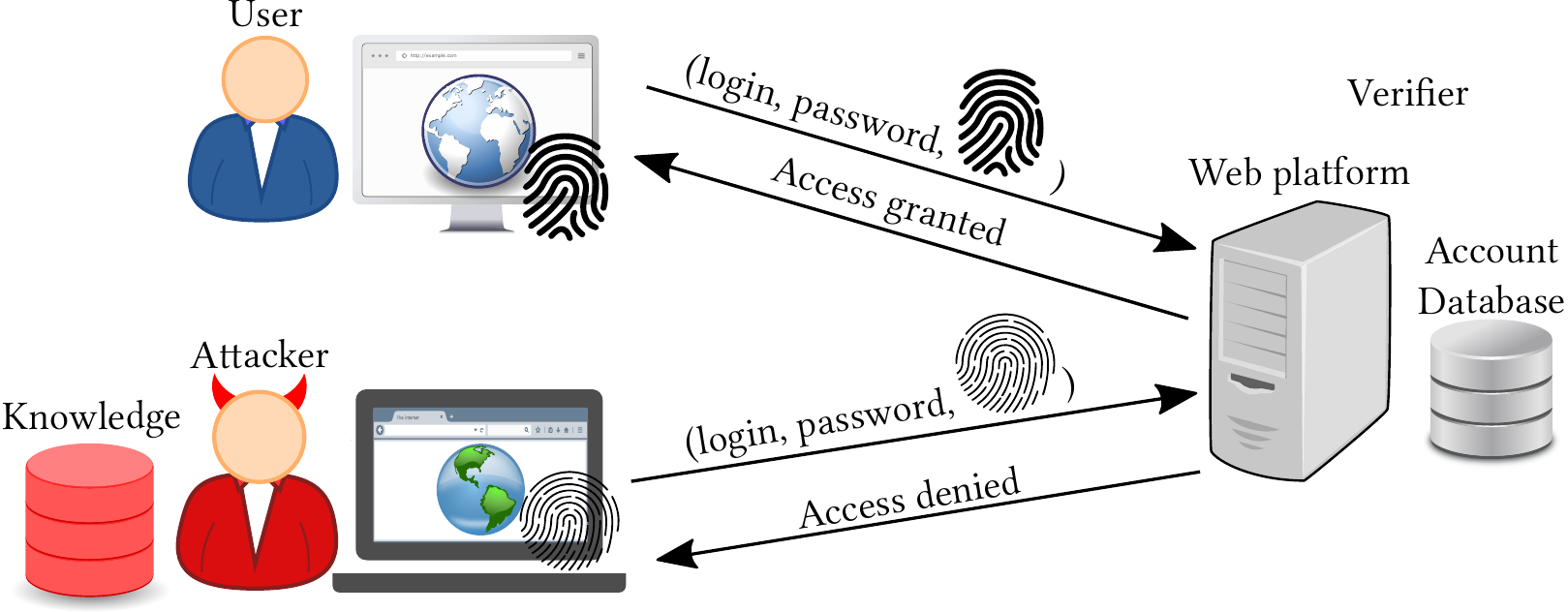}
    \caption{Example of a browser fingerprinting web authentication
      mechanism and a failed attack.}
    \label{fig:architecture-and-example}
    \Description[
      Example of a browser fingerprinting web authentication mechanism and a failed attack.
    ]{
      Example of a browser fingerprinting web authentication mechanism and a failed attack.
    }
  \end{figure}

  The two adversarial participants are the verifier and the attacker, as depicted in Figure~\ref{fig:architecture-and-example}.
  The verifier aims to protect the users of her web platform, using an authentication mechanism based on browser fingerprinting.
  The verifier stores the fingerprint of the usual browser of each user.
  On each login, the fingerprint of the browser in use is matched against the fingerprint that is stored for the claimed account.
  The attacker tries to impersonate the users by submitting specially crafted fingerprints.
  The aim of the verifier is to limit the reach of the attacker, also called \emph{sensitivity} below, which is measured as the proportion of impersonated users.
  To do so, she builds a fingerprinting probe that integrates one or more attributes, that are selected among the hundreds\footnote{
    Most attributes are properties accessed through the browser that are limited by its functionalities.
    Other attributes are items which presence are checked (e.g., the fonts~\cite{GAC16}, the extensions~\cite{SVS17}), or the computation of specific instructions (e.g., the HTML5 canvas~\cite{BMPT16}).
    These are limited by the available items or instructions, which can be large (e.g., more than $2^{154}$ for the canvas~\cite{LABN19}, nearly $30$ thousand detectable extensions~\cite{KISP20}).
  } that have been discovered over the years~\cite{ECK10, FE15, LRB16, AV16, SN17}.
  On the one hand, the addition of an attribute to the probe can strengthen the distinctiveness of browsers, hence reducing the sensitivity.
  On the other hand, each addition comes with a \emph{usability cost} that may render the probe impractical in an online authentication context.
  Indeed, each attribute consumes storage space (up to hundreds of kilobytes~\cite{BMPT16}), collection time (up to several minutes~\cite{MBYS11, MRHLSW13, NSY15, SYIHTCZ16, SYTT17, QF19}), and can increase the instability of the generated fingerprints~\cite{VLRR18}.
  For example, considering all of our attributes leads to a fingerprint taking $9.98$~seconds on average to collect, which is impractical for the user.
  Moreover, some attributes are strongly correlated together~\cite{AALG21}, and including them only increases the usability cost without reducing the sensitivity.
  Due to these correlations, picking attributes one by one independently may lead to poor sensitivity and usability scores.

  Previous works only consider the well-known attributes~\cite{ECK10, LRB16, GLB18}, remove the attributes of the lowest entropy~\cite{VLRR18}, iteratively pick the attribute of the highest weight (typically the entropy) until a threshold is reached~\cite{MEN11, KZW15, FE15, BRC16, HRA18, THS18}, or evaluate every possible set~\cite{FK12}.
  The entropy measures the skewness of the distribution of fingerprints or attribute values.
  As pointed out by Acar~\cite{ACA17}, it does not take the worst cases into account (i.e., the most common values that attackers can submit similarly to dictionary attacks on passwords~\cite{BON12}).
  Moreover, fingerprints cannot be compared identically like passwords due to their evolution through time.
  The attackers do not need to find the exact fingerprint of a victim, but one that is similar enough to deceive the verification mechanism.

  In this paper, we propose FPSelect, a framework that allows a verifier to select the attributes\footnote{
    We emphasize that the candidate attributes can contain dynamic attributes, which can be used to implement challenge-response mechanisms that resist fingerprint replay attacks~\cite{LABN19, REKP19}.
    We study nine instances of three dynamic attributes, which are the HTML5 canvas~\cite{BMPT16}, the WebGL canvas~\cite{MS12}, and audio fingerprinting methods~\cite{QF19}.
  } to include into her fingerprinting probe such that
  (1) the sensitivity against powerful attackers knowing the fingerprint distribution of the protected users (i.e., the worst-case distribution for the verifier) is bounded and the bound is set by the verifier, and
  (2) the usability cost\footnote{
    Any usability cost can be plugged (e.g, the privacy cost of including an attribute) provided that it is monotonic.
  } of collecting, storing, and using these attributes is close to being minimal.
  FPSelect is parameterized with the sensitivity requirement, the number of submissions that the attacker is deemed able to execute, and a representative sample of the fingerprints of the users.

  The problem could be solved by exploring exhaustively the space of the possible attribute sets, evaluating the sensitivity and the usability cost of each set.
  This is, however, infeasible as the number of attribute sets grows exponentially with the number of attributes\footnote{
    Obviously, this discards as well the manual selection of attributes.
  }.
  Moreover, we show below that the problem of finding the optimal attribute set is NP-hard.
  To the best of our knowledge, this is the first work that allows verifiers to dimension their fingerprinting probe in a sound manner, by quantifying the security level to reach, and selecting an attribute set that satisfies this level at a low usability cost.

  Our key contributions are the following:
  \begin{itemize}
    \item We formalize the \emph{attribute selection problem} that a verifier has to solve to dimension her probe.
      We show that this problem is NP-hard because it is a generalization of the Knapsack Problem.
      We define the model of the dictionary attacker, whose adversarial power depends on the knowledge of a fingerprint distribution.
      We propose a measure to quantify the sensitivity of a probe given a browser population and the number of fingerprints that the attacker is able to submit.
      We propose a measure of the usability cost that combines the size of the generated fingerprints, their collection time, and their instability.
    \item We propose a heuristic algorithm for selecting an attribute set that satisfies a higher bound on the sensitivity and reduces the usability cost.
      We express this as a search problem in the lattice of the power set of the candidate attributes.
      This algorithm is inspired by the Beam Search algorithm~\cite{JM09BeamSearch} and is part of the Forward Selection algorithms~\cite{SO14StepwiseRegression}.
    \item We evaluate the FPSelect framework on a real-life fingerprint dataset, and compare it with common attribute selection methods based on the entropy and the conditional entropy.
      We show experimentally that FPSelect finds attribute sets that have a lower usability cost.
      The attribute sets found by FPSelect generate fingerprints that are $12$ to $1,663$ times smaller, $9$ to $32,330$ times faster to collect, and with $4$ to $30$ times less changing attributes between two observations, compared to the candidate attributes and on average.
      Compared to the baselines, the attribute sets found by FPSelect generate fingerprints that are up to $97$ times smaller, are collected up to $3,361$ times faster, and with up to $7.2$ times less changing attributes between two observations, on average.
  \end{itemize}

  The rest of the paper is organized as follows.
  Section~\ref{sec:problem-statement} defines the attack model and the attribute selection problem.
  Section~\ref{sec:attribute-selection} describes the resolution algorithm and the proposed illustrative measures of sensitivity and usability cost.
  Section~\ref{sec:experimental-validation} provides the results obtained by processing our framework and the baselines on a real-life fingerprint dataset.
  Section~\ref{sec:discussion} discusses concrete usage of the framework.
  Section~\ref{sec:related-works} describes the works related to the attribute selection problem.
  Finally, Section~\ref{sec:conclusion} concludes.

\section{Problem Statement}
\label{sec:problem-statement}
  In this section, we first present the considered authentication mechanism that relies on browser fingerprinting.
  Then, we describe how we model the attacker given his knowledge and possible actions.
  Finally, we pose the attribute selection problem that we seek to solve, and provide an example to illustrate the problem.

  \subsection{Authentication Mechanism}
    We consider the architecture and the three participants depicted in Figure~\ref{fig:architecture-and-example}.
    The authentication mechanism is executed on a \emph{trusted web platform} and aims at authenticating \emph{legitimate users} based on various authentication factors, including their browser fingerprint (in addition to, e.g., a password).
    For the sake of precision, we focus on the browser fingerprint and ignore the other factors.

    A user is enrolled by providing his browser fingerprint to the verifier who stores it.
    During the authentication of a user, the fingerprint of the browser in use is collected by the fingerprinting probe of the verifier, and is compared with the fingerprint stored for the claimed account.
    If the collected fingerprint matches with the one stored, the user is given access to the account, and the stored fingerprint is updated to the newly collected one.
    The comparison is done using a \emph{matching function} (i.e., a similarity function between two fingerprints that authorizes differences), as fingerprints are known to evolve~\cite{ECK10, VLRR18, AAL21}.
    Any matching function can be used provided that it is monotonic (i.e., if two fingerprints match\footnote{
      We stress that the monotonic property does not depend on the attributes.
    } for an attribute set~$C$, they also match for any subset of~$C$).
    We explain in Section~\ref{sec:attribute-selection-problem} the need for the monotonicity requirement, and refer to Section~\ref{sec:implementation-of-the-matching-function} for an example of a matching function.
    We consider one browser per user and discuss the extension to multiple browsers in Section~\ref{sec:usage-of-multiple-browsers}.

  \begin{figure}
    \centering
    \includegraphics[width=\columnwidth]{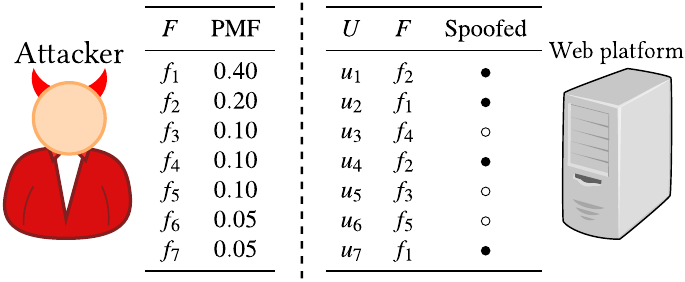}
    \caption{
      Example of an attacker instantiated with his knowledge of a probability mass function (PMF) over the fingerprints~$F$, and a web platform protecting a user population~$U$ with their fingerprint.
      We consider a limit of two submissions and a strict comparison between the fingerprints.
      The attack dictionary is composed of $f_1$ and $f_2$, resulting in the shown impersonated users.
    }
    \label{fig:attacker-instance}
    \Description[
      Example of an attacker instantiated with his knowledge of a probability mass function over the fingerprints, and a web platform protecting a user population using their fingerprint.
    ]{
      Example of an attacker instantiated with his knowledge of a probability mass function over the fingerprints, and a web platform protecting a user population using their fingerprint.
    }
  \end{figure}

  \subsection{Attack Model}
  \label{sec:attack-model}
    The high-level goal of the attacker is to impersonate legitimate users in a \emph{limited number of submissions} with the help of his knowledge, by forging a \emph{fingerprint attack dictionary} similarly to dictionary attacks on passwords~\cite{BON12}.
    Figure~\ref{fig:attacker-instance} illustrates the attack that we consider.
    It shows an attacker with his knowledge of a fingerprint distribution, a population of protected users with their fingerprint, and the impersonated users.
    We define the attacker model in terms of background knowledge and possible actions, which are provided below and described further in the following subsections.
    \begin{enumerate}
      \item The attacker cannot tamper with the web platform.
        \label{enum:cannot-tamper-web-platform}
      \item The attacker cannot tamper with, nor eavesdrop, the communication between the users and the web platform.
        \label{enum:cannot-eavesdrop-tamper-communications}
      \item The attacker knows the attributes of the probe.
        \label{enum:knows-implemented-attributes}
      \item The attacker knows a fingerprint distribution.
        \label{enum:know-distribution-fingerprints}
      \item The attacker can submit a limited number of arbitrary fingerprints.
        \label{enum:can-submit-arbitrary-fingerprints}
    \end{enumerate}

    \subsubsection{Background Knowledge}
      The attacker can retrieve the att\-ri\-bu\-tes of the fingerprinting probe (assumption~\ref{enum:knows-implemented-attributes}) by reverse-engi\-neer\-ing the probe (e.g., static or dynamic analysis of the probe~\cite{AAH18}, analysis of the network packets).

      The attacker knows the \emph{domain of the fingerprints}, and can infer a fingerprint distribution (assumption~\ref{enum:know-distribution-fingerprints}) from documentation~\cite{RFC7231UserAgent}, datasets\footnote{
        \url{https://www.henning-tillmann.de/en/2014/05/browser-fingerprinting-93-of-all-user-configurations-are-unique}
      }, or statistics\footnote{
        \url{http://carat.cs.helsinki.fi/statistics}
      } available online.
      He can also leverage phishing attacks~\cite{TLZBRIMCEMMPB17}, a pool of controlled browsers~\cite{PIPMIV19}, or stolen fingerprints~\cite{MAR20}.
      The weakest attacker is the one that lacks knowledge, and considers that the values of the attributes and fingerprints are uniformly distributed.
      His strategy is then to cover a space as large as possible of the fingerprint possibilities in the number of submissions authorized by the verifier.
      The strongest attacker is the one that manages to infer the exact fingerprint distribution among the users protected by the verifier.
      Additionally, our work can be easily extended to the attackers that partially know the fingerprints of targeted users\footnote{
        We do not consider the attackers that exactly know the fingerprint of the users they target (or their local configuration) because they are able to bypass trivially any fingerprinting authentication mechanism.
      } (see Section~\ref{sec:per-browser-attribute-sets}).

    \subsubsection{Actions}
      Tools exist for controlling the attributes\footnote{
        These tools are able to control both the fixed and the dynamic attributes.
      } that compose the fingerprint (assumption~\ref{enum:can-submit-arbitrary-fingerprints}), like Blink~\cite{LRB15} or Disguised Chromium Browser~\cite{BKST16}.
      Commercial solutions also exist, like AntiDetect\footnote{
        \url{https://antidetect.org}
      } or Multilogin\footnote{
        \url{https://multilogin.com}
      }.
      An attacker can also automatically alter the network packet that contains the fingerprint using tools like BurpSuite\footnote{
        \url{https://portswigger.net/burp}
      }.
      As these attacks are online guessing attacks~\cite{BON12, WZWYH16}, we assume that the attacker is limited to a number of submissions per user.

    \subsubsection{Attacker Instance}
      The verifier instantiates an attacker by his knowledge of a fingerprint distribution, and by the number of submissions to which he is limited, to measure his reach.

  \subsection{Attribute Selection Problem}
  \label{sec:attribute-selection-problem}
    The defense problem consists into selecting the attribute set that composes the fingerprinting probe, to resist against an instantiated attacker and minimize the usability cost.
    On the one hand, including an attribute can reduce the reach of an attacker -- called the \emph{sensitivity} and measured as the \emph{proportion of impersonated users} -- because it adds one more information to distinguish different browsers.
    On the other hand, it increases the \emph{usability cost} of the mechanism.
    For example, the fingerprints take more space to store, can take more time to collect, and can be more difficult to recognize due to the potentially induced instability.

    The \emph{Attribute Selection Problem} consists in finding the attribute set that provides the lowest usability cost and keeps the sensitivity below a threshold~$\alpha$ set by the verifier\footnote{
      The \emph{sensitivity threshold}~$\alpha$ is defined by the verifier according to her security requirements.
      These requirements depend on the type of website that is to protect (e.g., a bank, a forum) and the contribution of browser fingerprints (e.g., the only secondary authentication factor, an additional verification among others~\cite{SCCDVO19}).
    }.
    Let~$A$ denote the set of the candidate attributes.
    We consider an attribute set~${C \subseteq A}$, its usability cost~$\mathrm{c}(C)$, and its sensitivity~$\mathrm{s}(C)$.
    Any measure of usability cost and sensitivity can be plugged in FPSelect provided that it is \emph{monotonic}.
    Indeed, the usability cost is required to be \emph{strictly increasing} as we add attributes to an attribute set (e.g., the additional attributes are stored, which increases the storage cost).
    The sensitivity is required to be \emph{monotonically decreasing} as we add attributes to an attribute set\footnote{
      The monotonicity requirement of the matching function comes from the monotonicity requirement of the sensitivity.
      Indeed, if the matching function was not monotonic, adding an attribute could result in a loss of distinctiveness (i.e., it is harder for the matching function to distinguish two browsers) and consequently in an increase of the sensitivity.
    }.
    Indeed, adding an attribute to an attribute set should not higher the sensitivity because the added attribute either adds distinguishing information to the fingerprints or adds no information if it is strongly correlated with another attribute.
    For illustrative purposes, we propose measures of sensitivity and usability cost in Section~\ref{sec:attribute-selection}.
    The ASP is thus formalized as searching for ${\arg \min_{C \subseteq A} \{ \mathrm{c}(C) : \mathrm{s}(C) \leq \alpha \}}$.

  \subsection{Illustration of the Attribute Selection Problem}
    To illustrate the problem, we propose an example of a fingerprint distribution in Table~\ref{tab:fingerprint-dataset-example}.
    We consider an attacker who managed to infer the exact same distribution, and who is able to submit one fingerprint per user.
    If we solely include the \texttt{CookieEnabled} attribute which provides no distinctiveness, this attacker can impersonate every user by submitting the \texttt{True} value.
    Whereas including the \texttt{Language} and \texttt{Screen} attributes leads to unique fingerprints, which reduces the sensitivity to a sixth.
    Ignoring the \texttt{Cookie\-En\-abl\-ed} attribute reduces the usability cost without increasing the sensitivity.
    There is also an example of correlation.
    The \texttt{Timezone} and the \texttt{Language} attributes are the two most distinctive attributes, but including both does not improve the distinctiveness compared to considering \texttt{Language} alone.

    \begin{table}
      \centering
      \begin{tabular}{ccccc}
        \toprule
          User  & CookieEnabled & Language & Timezone & Screen \\
        \midrule
          $u_1$ & True          & fr       & -1       & 1080   \\
          $u_2$ & True          & en       & -1       & 1920   \\
          $u_3$ & True          & it       & 1        & 1080   \\
          $u_4$ & True          & sp       & 0        & 1920   \\
          $u_5$ & True          & en       & -1       & 1080   \\
          $u_6$ & True          & fr       & -1       & 1920   \\
        \bottomrule
      \end{tabular}

      \caption{
        Example of fingerprints shared by users.
      }
      \label{tab:fingerprint-dataset-example}
    \end{table}

\section{Attribute Selection Framework}
\label{sec:attribute-selection}
  This section is dedicated to the description of our attribute selection framework.
  First, we show that the Attribute Selection Problem (ASP) is NP-hard because it is a generalization of the Knapsack Problem (KP), and remark that the ASP can be seen as a lattice of partial KP.
  Second, and consequently, we propose a greedy heuristic algorithm for finding solutions to the problem.
  Finally, we propose illustrative measures for the sensitivity and the usability cost.

  \subsection{Similarity to the Knapsack Problem}
  \label{sec:asp-knapsack-problem}
    The \emph{Knapsack Problem} (KP)~\cite{KPP04} is a NP-hard problem that consists into fitting valued-items into a limited-size bag to maximize the total value.
    More precisely, given a bag of capacity~$W$ and $n$~items with their value~$v_i$ and their weight~$w_i$, we search for the item set that maximizes the total value and which total weight does not exceed~$W$.
    In this section, we show that the ASP is a generalization of the KP, therefore the ASP is NP-hard.
    We also provide a way to model the ASP as a lattice of partial KP.

    First, we remark that the ASP can be solved by picking attributes until we reach the sensitivity threshold, or by starting from the candidate attributes and removing attributes successively without exceeding the threshold.
    We consider the latter and start from the set~$A$ of the candidate attributes.
    The \emph{value} of an attribute set~$C$ is the cost reduction compared to the candidate attributes, formalized as ${v(C) = \mathrm{c}(A) - \mathrm{c}(C)}$.
    The \emph{value} of an attribute~$a$ is the cost reduction obtained when removing~$a$ from~$C$, which is formalized as ${v(a | C) = \mathrm{c}(C) - \mathrm{c}(C \setminus \{ a \})}$.
    The \emph{weight} of an attribute set~$C$ is its sensitivity with ${w(C) = \mathrm{s}(C)}$.
    The \emph{weight} of an attribute~$a$ is the additional sensitivity induced by the attribute removal, formalized as ${w(a | C) = \mathrm{s}(C \setminus \{ a \}) - \mathrm{s}(C)}$.
    The \emph{capacity}~$W$ is the maximum sensitivity allowed, hence ${W = \alpha}$.
    As we remove attributes, the value increases (i.e., the usability cost decreases), and the weight (i.e., the sensitivity) may increase.

    \begin{theorem}
    \label{thm:asp-knapsack-similarity}
      The Attribute Selection Problem is NP-hard.
    \end{theorem}
    \begin{proof}
      We consider a simple case where the attributes are not correlated.
      The weight and the value of the attribute~$a_i$ does not depend on the attributes already included in the probe, and is simply defined as~$w_i$ and~$v_i$.
      We obtain a Knapsack Problem consisting into picking the attributes to remove from~$A$, to maximize the total value and keep the weight under the threshold~$W$.
      The ASP is therefore a generalization of the KP with relative weights and costs, making it at least as hard as the KP which is NP-hard.
      The Attribute Selection Problem is therefore NP-hard.
    \end{proof}

    \subsubsection{The Attribute Selection Problem as a Lattice of Partial Knapsack Problems}
      The Attribute Selection Problem can be modeled as a lattice of partial Knapsack Problems (KP).
      We consider the deletive way that starts from the set~$A$ of the candidate attributes and removes attributes without exceeding the threshold.
      The initial partial KP consists into picking attributes from~$A$ to increase the value and keep the weight under~$W$.
      The value and weight of each attribute~${a \in A}$ is ${v(a | A)}$ and ${w(a | A)}$.
      Once we pick an attribute~$a_p$, a new partial KP arises: the item set is~${A \setminus \{ a_p \}}$, the capacity is~${W - w(a_p | A)}$, and the value and weight of each attribute~${a \in A \setminus \{ a_p \}}$ is now ${v(a | A \setminus \{ a_p \})}$ and ${w(a | A \setminus \{ a_p \})}$.
      Recursively, it holds for any set~$R$ of attributes to remove.
      The item set is then~${A \setminus R}$, the capacity is~${W - w(R)}$, and the value and weight of each attribute~${a \in A \setminus R}$ are ${v(a | A \setminus R)}$ and ${w(a | A \setminus R)}$.
      Following this, we are given a lattice\footnote{
        This can be seen as a tree, but some paths lead to the same node.
        Indeed, removing the attributes $a_1$ then $a_2$ from $A$ leads to the same partial problem as removing $a_2$ then $a_1$.
      } of partial KP to solve recursively, each node being a partial solution~$R$, until we reach unfeasible problems (i.e., empty set of items, no more item can fit) and find a final solution among the partial solutions that reach this limit.

  \subsection{Lattice Model and Resolution Algorithm}
    In this section, we present how we model the possibility space as a lattice of attribute sets, and describe the greedy heuristic algorithm to approximately solve the ASP.

    \begin{figure}
      \centering
      \includegraphics[width=0.8\columnwidth]{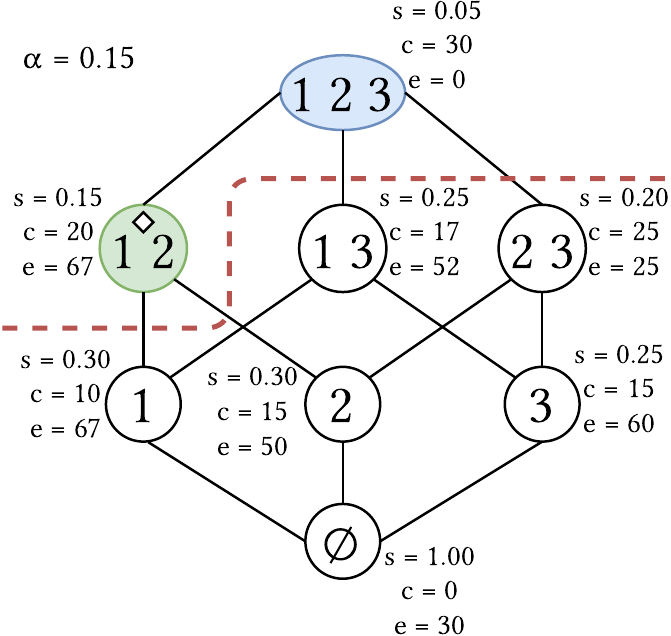}
      \caption{
        Example of a lattice of attribute sets, with their cost~$\mathrm{c}$, their sensitivity~$\mathrm{s}$, and their efficiency~$\mathrm{e}$.
        The blue node satisfies the sensitivity, the white nodes do not, and the green node with a diamond satisfies the sensitivity and minimizes the cost.
        The red line is the satisfiability frontier.
      }
      \label{fig:lattice-example}
      \Description[
        Example of a lattice of attribute sets, with their cost, their sensitivity, and their efficiency.
      ]{
        Example of a lattice of attribute sets, with their cost, their sensitivity, and their efficiency.
      }
    \end{figure}

    \subsubsection{Lattice Model}
    \label{sec:lattice-model}
      The elements of the lattice are the \emph{subsets} of~$A$ ($A$ included) and the order is the \emph{subset relationship} so that ${C_i \prec C_j}$ if, and only if, ${C_i \subset C_j}$.
      The efficiency of an attribute set~$C$ is the ratio between its cost reduction (i.e., ${c(A) - c(C)}$) and its sensitivity.
      Figure~\ref{fig:lattice-example} shows an example of such lattice.
      The \emph{satisfiability frontier} represents the transition between the attribute sets that satisfy the sensitivity threshold, and those that do not.
      The attribute sets just above this frontier satisfy the sensitivity threshold at a lower cost than any of their supersets.
      They comprise the solution found by our resolution algorithm and the optimal solution to the problem.

      The sensitivity and the cost are bounded.
      The lower bound is located at the empty set, which has a sensitivity of $1.0$ and a null usability cost.
      It is equivalent to not using browser fingerprinting at all.
      On the other end, the set composed of the candidate attributes~$A$ is a superset of every attribute set, and provides the lowest sensitivity and the highest usability cost.
      If $A$ does not satisfy the sensitivity threshold, there is no solution as any other subset has a higher or equal sensitivity.

    \subsubsection{Greedy Algorithm}
      We propose the greedy heuristic algorithm presented in Algorithm~\ref{alg:asp-greedy} to find good solutions to the Attribute Selection Problem.
      It consists into a bottom-up exploration of the lattice by following $k$~paths until reaching the satisfiability frontier.
      The higher $k$ is, the larger is the explored space, but the higher is the computing time.
      This algorithm is inspired by the model of the ASP as a lattice of partial Knapsack Problems, and by the Beam Search algorithm~\cite{JM09BeamSearch}.
      The similarity with the latter lies in the successive expansion of a limited number of nodes, that are chosen according to an order between partial solutions.
      The order is the efficiency in our case.
      Our proposed algorithm is part of the Forward Selection algorithms~\cite{SO14StepwiseRegression}, as it iteratively picks attributes according to a criterion, and takes into account those already chosen.
      However, the proposed algorithm provides the ability to explore several sub-solutions instead of a single one, includes pruning methods that help reduce the computing time, and stops when it reaches the satisfiability frontier instead of when the criterion is not statistically improved.

      \begin{algorithm}[t]
        \caption{Greedy algorithm to find good solutions to the Attribute Selection Problem.}
        \label{alg:asp-greedy}

        \KwData{The candidate attributes~$A$, the sensitivity threshold~$\alpha$, the number of explored paths~$k$.}
        \KwResult{The attribute set of the explored paths that satisfies the sensitivity threshold at the lowest cost.}

        $c_{min}, T, I \gets \inf, \varnothing, \varnothing$ \\
        $S \gets$ a collection of $k$ empty sets \\

        \If{$\mathrm{s}(A) > \alpha$}{
          Quit as no solution exists
        }

        \While{$S$ is not empty}{
          $E \gets \{
            {C = S_i \cup \{ a \}} :
              {\forall S_i \in S},
              {\forall a \in A \setminus S_i},
              {\nexists C' \in T \cup I},
              {C' \subset C}
          \}$ \\

          $S \gets \varnothing$ \\

          \For{$C \in E$}{
            \If{$s(C) \leq \alpha$}{
              $T \gets T \cup \{ C \}$ \\
              $c_{min} \gets c(C)$ if $c(C) < c_{min}$ \\
            }
            \ElseIf{$c(C) < c_{min}$}{
              $S \gets S \cup \{ C \}$ \\
            }
            \Else{
              $I \gets I \cup \{ C \}$ \\
            }
          }

          $S \gets$ the~$k$ most efficient attribute sets~$C$ of~$S$ according to $\frac{\mathrm{c}(A) - \mathrm{c}(C)}{\mathrm{s}(C)}$
        }

        \Return{$\arg \min_{C \in T} \mathrm{c}(C)$}
      \end{algorithm}

      \paragraph{Algorithm Working}
        Algorithm~\ref{alg:asp-greedy} works by exploring $k$~paths of the lattice.
        It starts from the empty set and stops when every path reaches the satisfiability frontier.
        The collection~$S$ holds the attribute sets to expand and is initialized to $k$~empty sets.
        At each stage, the attribute sets to explore are stored in the collection~$E$.
        They consist of each~${S_i \in S}$ with one additional attribute.
        The cost and the sensitivity of each attribute set~${C \in E}$ is then measured.
        If $C$~satisfies the sensitivity threshold~$\alpha$, it is added to the collection~$T$ of the attribute sets that reach the satisfiability frontier, otherwise it is added to the collection~$S$ of the attribute sets to expand.
        Finally, the collection~$S$ is updated to only hold the $k$~most efficient attribute sets.
        The efficiency of an attribute set is the ratio between its gain (i.e., the cost reduction compared to the candidate attributes) and its sensitivity.
        All this process is repeated until~$S$ is empty, when all the $k$ paths have reached the satisfiability frontier.
        The solution is then the attribute set of the lowest cost in~$T$.

      \paragraph{Pruning Methods}
        Three properties allow us to reduce the number of attribute sets that are explored.
        First, we hold the minimum cost~$c_{min}$ of the attribute sets~$T$ that satisfy the sensitivity.
        Any explored attribute set that has a cost higher than~$c_{min}$ is not added to the collection~$S$ of those to explore.
        Indeed, this attribute set does not provide the lowest cost, nor do its supersets.
        Then, during the expansion of two attribute sets $S_i$ and $S_j$ of the same size, if $S_i$ satisfies the sensitivity and $S_j$ does not, they can have a common superset~$S_l$.
        In this case, $S_l$ does not need to be explored as it costs more than~$S_i$.
        We store $S_i$ in $I$ so that we can check if an attribute set~$C$ has a subset in $I$, in which case we do not explore $C$.
        The same holds if $S_i$ costs less than $c_{min}$ and $S_j$ costs more than $c_{min}$.

      \paragraph{Algorithm Complexity}
        Starting from the empty set, we have $n$~supersets composed of one more attribute.
        From these $n$~supersets, we update~$S$ to hold at most $k$~attribute sets.
        The attribute sets~${S_i \in S}$ are now composed of a single attribute, and each~$S_i$ has ${n-1}$~supersets composed of one additional attribute.
        At any stage, we have at most ${k n}$~attribute sets to explore.
        This process is repeated at most $n$~times, as we can add at most $n$~attributes, hence the \emph{computational complexity} of the Algorithm~\ref{alg:asp-greedy} is of~${\mathcal{O}(k n^2 \omega)}$, with $\omega$ being the computational complexity of the measures of usability cost and sensitivity of an attribute set.
        The collection~$E$ contains at most ${k n}$~attribute sets (at most $n$~supersets for each ${S_i \in S}$).
        The collections $S$, $T$, and $I$ can contain more sets, but are bounded by the number of explored nodes which is ${k n^2}$.
        The \emph{memory complexity} of the Algorithm~\ref{alg:asp-greedy} is then of~${\mathcal{O}(k n^2)}$.

    \paragraph{Example}
      Table~\ref{tab:example-greedy-lattice-exploration} displays an example of the execution of Algorithm~\ref{alg:asp-greedy} on the lattice presented in Figure~\ref{fig:lattice-example}, with the sensitivity threshold~${\alpha=0.15}$ and the number of explored paths~${k=2}$.
      The stage~$i$ corresponds to the state at the end of the $i$-th \texttt{while} loop.
      Initially, the collection~$S$ is~${\{ \varnothing, \varnothing \}}$.
      At stage~$1$, the two most efficient attribute sets of~$E$ are ${\{ 1 \}}$ and ${\{ 3 \}}$, which are stored into~$S$.
      At stage~$2$, we assume that the attribute set~${\{ 1, 2 \}}$ is measured first as there is no order among~$E$.
      In this case, this attribute set is added to the collection~$T$, and the minimum cost is now~$20$.
      The attribute set~$\{ 1, 3 \}$ is then added to~$S$, but $\{ 2, 3 \}$ is not as it has a higher cost than the minimum cost.
      At stage~$3$, the attribute set~${\{ 1, 2, 3 \}}$ is not added to the collection~$E$ as it is a superset of one attribute set of the collection~$T$.
      The final solution is the less costly attribute set of~$T$, which is~${\{ 1, 2 \}}$ in this case, and happens to be the optimal solution.

      \begin{table}
        \centering
        \begin{tabular}{lllll}
          \toprule
            Stage & $E$   & $T$        & $S$ \\
          \midrule
            1    &  $\{ \{ 1 \}, \{ 2 \}, \{ 3 \} \}$  &  $\{ \}$  &  $\{ \{ 1 \}, \{ 3 \} \}$  \\
            2    &  $\{ \{ 1, 2 \}, \{ 1, 3 \}, \{ 2, 3 \} \}$  &  $\{ \{ 1, 2 \} \}$  &  $\{ \{ 1, 3 \} \}$  \\
            3    &  $\{ \}$        &  $\{ \{ 1, 2 \} \}$  &  $\{ \}$      \\
          \bottomrule
        \end{tabular}

        \caption{
          Example of the execution of Algorithm~\ref{alg:asp-greedy} on the lattice of Figure~\ref{fig:lattice-example}, with the sensitivity threshold~${\alpha=0.15}$ and the number of explored paths~${k=2}$.
          Stage~$i$ is the state at the end of the $i$-th \texttt{while} loop.
        }
        \label{tab:example-greedy-lattice-exploration}
      \end{table}

  \subsection{Illustrative Measures of Sensitivity and Usability Cost}
    In this section, we illustrate a sensitivity measure as the proportion of impersonated users given the strongest attacker of our model that knows the fingerprint distribution among the protected users.
    We also illustrate a usability cost measure according to the fingerprints generated by a fingerprinting probe on a browser population.

    \subsubsection{Sensitivity Measure}
    \label{sec:sensitivity}
      We measure the sensitivity of a given attribute set according to an instantiated attacker and a population of users sharing browser fingerprints.
      The attacker knows the fingerprint distribution of the protected users, and submits orderly the most probable fingerprints, until reaching the threshold on the number of submissions.
      The illustrative \emph{sensitivity measure} evaluates the proportion of users that are impersonated considering the matching function.

      From an attribute set~$C$, we retrieve the fingerprint domain~$F_C$ such that ${F_C = \prod_{a \in C} \mathrm{domain}(a)}$, with ${\mathrm{domain}(a)}$~being the domain of the attribute~$a$ and $\prod$~being the Cartesian product.
      We denote $F_{A}$ the fingerprints when considering the set~$A$ of the candidate attributes.
      We denote $U$ the set of the users that are protected by the verifier.
      The set ${\mathcal{M} = \{ (u, f) : u \in U, f \in F_{A} \}}$ represents the mapping from the users to their fingerprint, so that the user~$u$ has the fingerprint~$f$ stored.

      We denote ${\mathrm{project(f, C)}}$ the function that projects the fingerprint ${f \in F_{C'}}$ from the set of attributes~$C'$ to the set of attributes~$C$, with the requirement that ${C \subseteq C'}$.
      Finally, the function denoted ${\mathrm{dictionary}(p, F_C, \beta)}$ retrieves the $\beta$-most probable fingerprints of $F_C$ given the probability mass function~$p$.
      We note that it is trivial to retrieve the distribution of the fingerprints composed of any attribute subset~${C \subset A}$ from the distribution of the fingerprints composed of the candidate attributes~$A$.

      We denote ${f[a]}$ the value of the attribute~$a$ for the fingerprint~$f$, and ${f[a] \approx^a g[a]}$ the matching between the value of the attribute~$a$ for the stored fingerprint~$f$ and the submitted fingerprint~$g$.
      It is true only if ${f[a]}$ matches with ${g[a]}$, meaning that ${g[a]}$ is deemed a legitimate evolution of ${f[a]}$.
      Finally, we define the set of the matching functions of each candidate attribute as ${\Phi = \{ \approx^a : a \in A \}}$.

      We measure the sensitivity as the proportion of impersonated users among a population of protected users, against the attacker that knows the fingerprint distribution among them, using Algorithm~\ref{alg:sensitivity-measure}.
      The illustrative sensitivity measure is \emph{mo\-no\-to\-nic} as demonstrated in Appendix~\ref{app:sensitivity-usability-cost-monotonicity-demonstrations}.

      The \emph{number of submissions} is defined by the verifier according to her rate limiting policy~\cite{GSD18} (e.g., blocking the account after three failed attempts).
      This limit could be set to $1$ as a user cannot mistake his browser fingerprint.
      However, a user can browse from a new or a public browser, and taking preventive action on this sole motive is unreasonable.

      \begin{algorithm}
        \caption{Illustrative sensitivity measure.}
        \label{alg:sensitivity-measure}

        \KwData{The attribute set~$C$, the limit on the number of submissions~$\beta$, the mapping~$\mathcal{M}$ from the users to their browser fingerprint, the probability mass function~$p$, and the set~$\Phi$ of matching functions.}
        \KwResult{The proportion of impersonated users.}

        $R \gets \{ \}$ \\
        $F_C \gets $ the fingerprint domain when considering~$C$ \\
        $V \gets \mathrm{dictionary}(p, F_C, \beta)$ \\

        \ForAll{$(u, f^{*}) \in \mathcal{M}$}{
          $f \gets \mathrm{project}(f^{*}, C)$ \\
          \If{$\exists g \in V$ st. $\forall a \in C, f[a] \approx^a g[a]$}{
            $R \gets R \cup \{ u \}$ \\
          }
        }

        \Return{$\frac{\mathrm{card}(R)}{\mathrm{card}(U)}$}
      \end{algorithm}

    \subsubsection{Usability Cost Measure}
    \label{sec:usability-cost}
      There is no off-the-shelf measure of the usability cost of the attributes (e.g., the \texttt{UserAgent} HTTP header~\cite{RFC7231UserAgent} has no specified size, collection time, nor change frequency).
      This cost also depends on the fingerprinted population (e.g., mobile browsers generally have fewer plugins than desktop browsers, resulting in smaller values for the list of plugins~\cite{AALG21}).
      As a result, we design an illustrative cost measure that combines three sources of cost (i.e., space, time, and instability), which is computed by the verifier on her fingerprint dataset.
      The \emph{fingerprint dataset} used to measure the costs is denoted ${D = \{ (b, f) : b \in B, f \in F_{A} \}}$, with $B$ being the set of observed browsers.

      The \emph{memory cost} is measured as the average fingerprint size.
      The attribute values are stored and not compressed into a single hash, which is necessary due to their evolution through time.
      We denote ${\mathrm{mem}(C, D)}$ the \emph{memory cost} of the attribute set~$C$, and ${\mathrm{size}(x)}$ the size of the value~$x$.
      The memory cost is defined as
      \begin{equation}
        \mathrm{mem}(C, D) =
          \frac{1}{\mathrm{card}(D)}
          \sum_{(b, f) \in D} \sum_{a \in C} \mathrm{size}(f[a])
      \end{equation}

      The \emph{temporal cost} is measured as the average fingerprint collection time, and takes into account the asynchronous collection of some attributes.
      Although attributes can be collected asynchro\-nous\-ly, some require a non-negligible collection time (e.g., the dynamic attributes~\cite{MS12, QF19}).
      We denote ${\mathrm{time}(C, D)}$ the temporal cost of the attribute set~$C$.
      Let~$A_{\mathrm{seq}}$ be the set of the sequential attributes, and $A_{\mathrm{async}}$ the set of the asynchronous attributes, so that we have ${C = A_{\mathrm{seq}} \cup A_{\mathrm{async}}}$.
      Let~${\mathrm{t}(b, f[a])}$ be the collection time of the attribute~$a$ for the fingerprint~$f$ collected from the browser~$b$.
      The temporal cost is defined as
      \begin{equation}
        \begin{split}
          \mathrm{time}(C, D) =
            \frac{1}{\mathrm{card}(D)}
            \sum_{(b, f) \in D}
            \max(
              & \{ \mathrm{t}(b, f[a]) : a \in A_{\mathrm{async}} \} \\
              & \cup \{ \sum_{s \in A_{\mathrm{seq}}} \mathrm{t}(b, f[s]) \}
            )
        \end{split}
      \end{equation}

      The \emph{instability cost} is measured as the average number of changing attributes between two consecutive observations of the fingerprint of a browser.
      We denote ${\mathrm{ins}(C, D)}$ the instability cost of the attribute set~$C$.
      We denote ${\mathcal{C}(D)}$ the non-empty set of the consecutive fingerprints coming from the same browser in the dataset~$D$, and ${\delta(x, y)}$~the Kronecker delta being $1$~if $x$~equals~$y$ and $0$~otherwise.
      The instability cost is defined as
      \begin{equation}
        \mathrm{ins}(C, D) =
          \frac{1}{\mathrm{card}(\mathcal{C}(D))}
          \sum_{(f, g) \in \mathcal{C}(D)} \sum_{a \in C}
            \delta(f[a], g[a])
      \end{equation}

      The three dimensions of the cost are weighted by a three-dim\-en\-sion\-al \emph{weight vector} denoted ${\gamma = [\gamma_1, \gamma_2, \gamma_3]}$ such that the weights are strictly positive numbers.
      The verifier tunes these weights according to her needs (e.g., allowing fingerprints to be more unstable, but requiring a shorter collection time).
      She can do this by defining an equivalence between the three dimensions (e.g., one millisecond of collection time is worth ten kilobytes of size), and setting the weights so that these values amount to the same quantity in the total cost.
      For a concrete example, we refer to Section~\ref{sec:usability-cost-measure}.

      Finally, we denote ${\mathrm{cost}(C, D)}$ the cost of the attribute set~$C$ given the fingerprint dataset~$D$.
      The illustrative usability cost measure is \emph{monotonic} as demonstrated in Appendix~\ref{app:sensitivity-usability-cost-monotonicity-demonstrations}, and is formalized as
      \begin{equation}
        \mathrm{cost}(C, D) = \gamma \cdot [
          \mathrm{mem}(C, D), \mathrm{time}(C, D), \mathrm{ins}(C, D)
        ]^\intercal
      \end{equation}

\section{Experimental Validation}
\label{sec:experimental-validation}
  In this section, we describe the experiments that we perform to validate our framework.
  We begin by presenting the fingerprint dataset that is used, and describing how the usability cost and the matching function are implemented.
  Then, we present the results of the attribute selection framework executed with different parameters, and compare them with the results of the common baselines.
  The experiments were performed on a desktop computer with $32$GB of RAM and $32$~cores running at $2$GHz.

  \subsection{Fingerprint Dataset}
    The fingerprint dataset used in this work is the same dataset as the one studied by Andriamilanto et al. in~\cite{AAL21, AALG21}.
    It was collected from December~$7$, $2016$, to June~$7$, $2017$, during a real-life experiment in which the authors integrated a fingerprinting probe to two pages of one of the $15$~most visited websites in France.
    We refer to their studies~\cite{AAL21, AALG21} that provide an in-depth analysis of this dataset, a comprehensive description of the fingerprint collection, a precise description of the preprocessing steps that include a cookie resynchronization process similar to~\cite{ECK10}, and an exhaustive list of the attributes with their properties.
    In a nutshell, the preprocessed dataset contains $5,714,738$~entries (comprising identical fingerprints for a given browser if interleaved\footnote{
      A browser can present interleaved fingerprints~\cite{AALG21} like $a$, $b$, then $a$ again.
      They typically come from a switch between two environments, like a laptop to which an external screen is plugged and unplugged.
      This browser has $3$ entries, but only has $2$ fingerprints ($a$ and $b$) to avoid over counting.
      These interleaved fingerprints are held when measuring the instability cost.
    }) and $4,145,408$~fingerprints (no identical fingerprint counted for the same browser), that are collected from $1,989,366$~browsers.
    The instability is evaluated from $3,725,373$~pairs of consecutive fingerprints, coming from the $27.53$\% of browsers that have multiple entries.
    The fingerprints are composed of $253$~candidate attributes, of which $49$ attributes are completely correlated with another one~\cite{AALG21}, so that knowing the value of the other attribute allows to completely infer their value.
    Although these attributes are correlated, we cannot simply remove them as the attributes present dissimilar costs that also depend on the attributes that are already considered.

  \begin{figure*}
    \minipage{\columnwidth}
      \centering
      \includegraphics[width=\columnwidth]{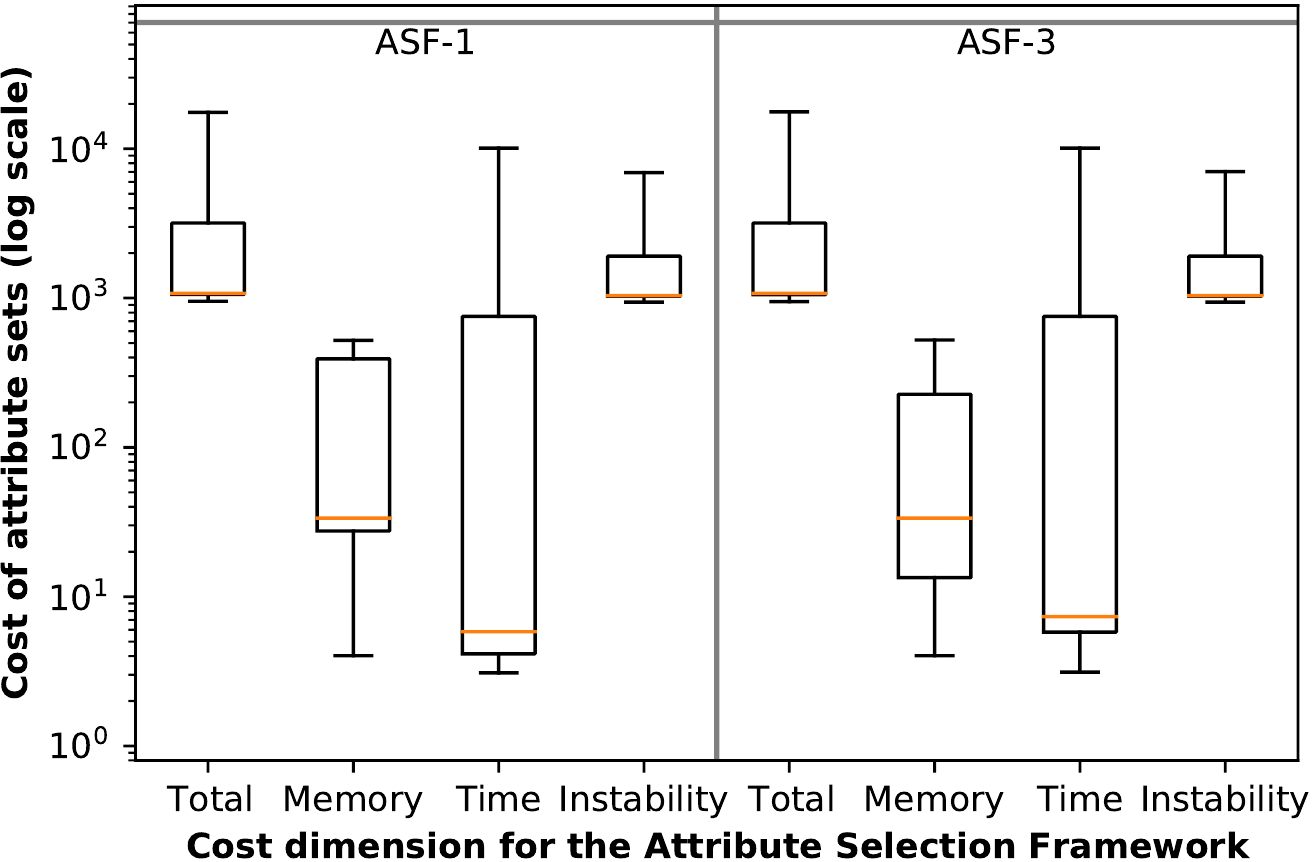}
    \endminipage
    \hfill
    \minipage{\columnwidth}
      \centering
      \includegraphics[width=\columnwidth]{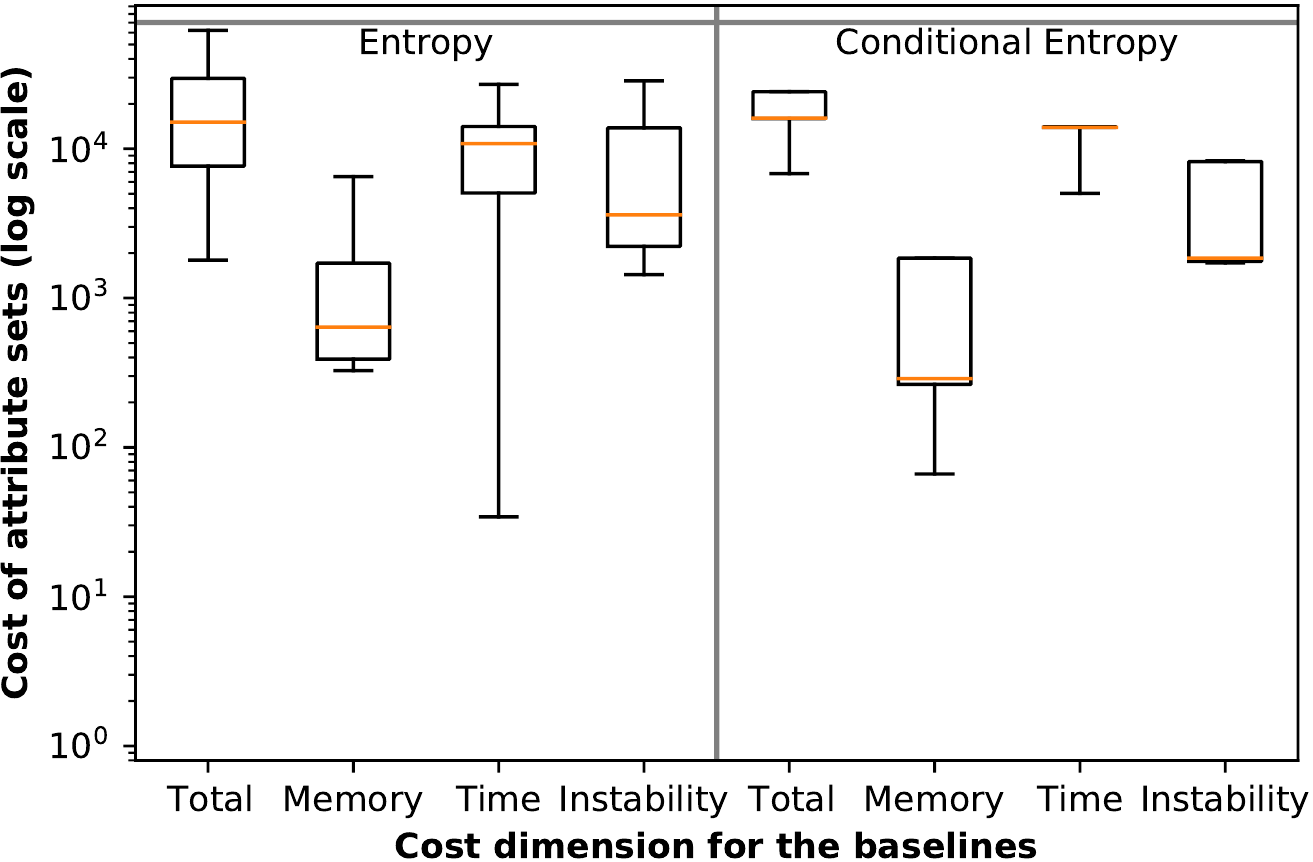}
    \endminipage
    \caption{
      Cost of the attribute sets found by the ASF with $1$ explored path (ASF-1), the ASF with $3$ explored paths (ASF-3), the entropy, and the conditional entropy.
      The costs are in points, so that $10,000$ additional points increases the size of fingerprints by $10$~kilobytes, their collection time by $1$~second, or the number of changing attributes between observations by $1$~attribute, on average.
      A solution exists for $9$ of the $12$ cases.
      The gray horizontal line is the cost when considering all the candidate attributes.
    }
    \label{fig:asf-cost-breakdown}
    \Description[
      Cost of the attribute sets found by the ASF-1, the ASF-3, the entropy, and the conditional entropy.
    ]{
      Cost of the attribute sets found by the ASF-1, the ASF-3, the entropy, and the conditional entropy.
    }
  \end{figure*}

  \subsection{Instantiation of the Experiments}
    In this section, we present the instantiation of the parameters for the experiments.
    We first describe how the verifier and the attacker are instantiated by presenting the chosen user population, sensitivity thresholds, and number of submissions.
    Then, we detail the implementation of the usability cost measure and the matching function between fingerprints, alongside the value of the parameters or weights that they use.

    \subsubsection{Verifier Instantiation}
      On the verifier side, we simulate a \emph{user population} by randomly sampling $30,000$~browsers from the first month of the experiment to represent a medium-sized website.
      The observed fingerprint is considered as the fingerprint stored for the user who owns the browser.
      We configure the resolution algorithm to have $1$ and $3$ \emph{explored paths} to compare the gain achieved by a larger explored space.
      We call \emph{ASF-1} and \emph{ASF-3} our attribute selection method with respectively $1$ and $3$ explored paths.
      We consider the set of \emph{sensitivity thresholds} $\{ 0.001, 0.005, 0.015, 0.025 \}$.
      Bonneau et al.~\cite{BHOS12} defined the resistance against online attacks as a compromise of $1$\% of accounts after a year when $10$ guesses per day are allowed.
      Hayashi et al~\cite{HDCP08} estimated that $0.001$ is equivalent to a random guess of four-digit.
      However, to the best of our knowledge, no standard value exists.
      Hence, we make the choice of these values starting from $0.001$ and going to $0.025$ to obtain a range from a strict security requirement to one that is less strict.
      We admit that $0.025$ ($2.5$\%) is already high, but it is close to the proportion of users that share the $10$ most common passwords in previously leaked datasets~\cite{WZWYH16}.

    \subsubsection{Attacker Instantiation}
      On the attacker side, an instance is parameterized with the number of fingerprints~$\beta$ that he can submit, and his knowledge over the fingerprint distribution.
      We consider the strongest attacker of our attack model that knows the fingerprint distribution among the user population.

      We consider the set of number of submissions $\{ 1, 4, 16 \}$.
      To the best of our knowledge, no standard value exists.
      The choice of $1$ is for a strict rate limiting policy that blocks the account on any failure and asks the user to change his password.
      The choice of $4$ is for a policy that would require a CAPTCHA after $3$ failed attempts, and would perform the blocking and password change after the fourth failed attempt.
      Finally, the choice of $16$ is for a policy that would let more attempts before performing the blocking and password change.
      The chosen values are close to the number of submissions allowed into policies enforced in real life~\cite{GSD18}, and the ones estimated as reasonable values against online guessing attacks~\cite{BON12}.

    \subsubsection{Implementation of the Usability Cost Measure}
    \label{sec:usability-cost-measure}
      The implemented \emph{usability cost function} measures the memory in bytes (a character per byte), the time in milliseconds, and the instability as the average number of changing attributes between the consecutive fingerprints.
      We configure the three-dimensional \emph{weight vector} to the values ${\gamma=[1; 10; 10,000]}$ to have an equivalence between $10$~kilobytes, $1$~second, and $1$~changing attribute on average, which are all equal to $10,000$~points.
      Table~\ref{tab:candidate-set-costs-min-mean-max-attribute-cost} displays the cost of the $253$ candidate attributes, together with the minimum, the average, and the maximum cost of a single attribute for each cost dimension.

      \begin{table}
        \centering
        \begin{tabular}{lllll}
          \toprule
            Value     & Cost (pts) & Memory (B) & Time (s) & Inst. (chgs)   \\
          \midrule
            Candidate & 134,270    & 6,114      & 9.98     & 2.83           \\
          \midrule
            Max. cost & 99,846     & 1,102      & 9.98     & 0.51           \\
            Avg. cost & 8,794      & 26         & 0.87     & 1.13 $10^{-2}$ \\
            Min. cost & 1          & 1          & 0        & 0.00           \\
          \bottomrule
        \end{tabular}

        \caption{
          The cost of the $253$ candidate attributes, together with the maximum, the average, and the minimum cost of a single attribute for each cost dimension.
        }
        \label{tab:candidate-set-costs-min-mean-max-attribute-cost}
      \end{table}

    \subsubsection{Implementation of the Matching Function}
    \label{sec:implementation-of-the-matching-function}
      The implemented \emph{matching function} checks that the distance between the attribute values of the submitted fingerprint and the stored fingerprint is below a threshold.
      Similarly to previous studies~\cite{ECK10, KCFLLL17, VLRR18}, we consider a distance measure that depends on the type of the attribute.
      The minimum edit distance~\cite{JM09MinimumEditDistance} is used for the textual attributes, the Jaccard distance~\cite{WWWLY16} is used for the set attributes, the absolute difference is used for the numerical attributes, and the reverse of the Kronecker delta (i.e., ${1 - \delta(x, y)}$) is used for the categorical attributes.
      The dynamic attributes (e.g., HTML5 canvas) are matched identically (i.e., using a threshold of $1$) as they serve the challenge-response mechanism.
      More complex matching functions exist (e.g., based on rules and machine learning~\cite{VLRR18}).
      They can be integrated to the framework as long as they are monotonic\footnote{
        A matching function is monotonic if two fingerprints that match for an attribute set~$C$ also match for any subset of~$C$.
      }.

      The \emph{distance threshold} for each attribute is set using Support Vector Machines (SVM)~\cite{HDOPS98} and the following methodology.
      First, we split our dataset into $6$~samples (one for each month) and extract the positive and negative classes.
      They respectively consist into the consecutive fingerprints of a browser, and two randomly picked fingerprints of different browsers.
      We assume that a user spends at most one month between each connection, and otherwise would accept to process a heavier fingerprint update process.
      Then, for each attribute, we train an SVM model on the two classes of each monthly sample, and extract the threshold from the resulting hyperplane.
      Finally, we compute the average of the $6$~obtained thresholds to get the distance threshold for each attribute.

    \subsubsection{Baselines}
      We compare our method with common attribute selection methods.
      The entropy-based method~\cite{MEN11, KZW15, BRC16, HRA18} consists into picking the attributes of the highest entropy until reaching an arbitrary number of attributes.
      The method based on the conditional entropy~\cite{FE15} consists into iteratively picking the most entropic attribute according to the attributes that are already chosen, and re-evaluating the conditional entropy of the remaining attributes at each step, until an arbitrary number of attributes is reached.
      Instead of limiting to a given number of attributes, we pick attributes until the obtained attribute set satisfies the sensitivity threshold.
      For simplification, we call \emph{entropy} and \emph{conditional entropy} the attribute selection methods that rely on these two metrics.

  \begin{figure*}
    \minipage{\columnwidth}
      \centering
      \includegraphics[width=\columnwidth]{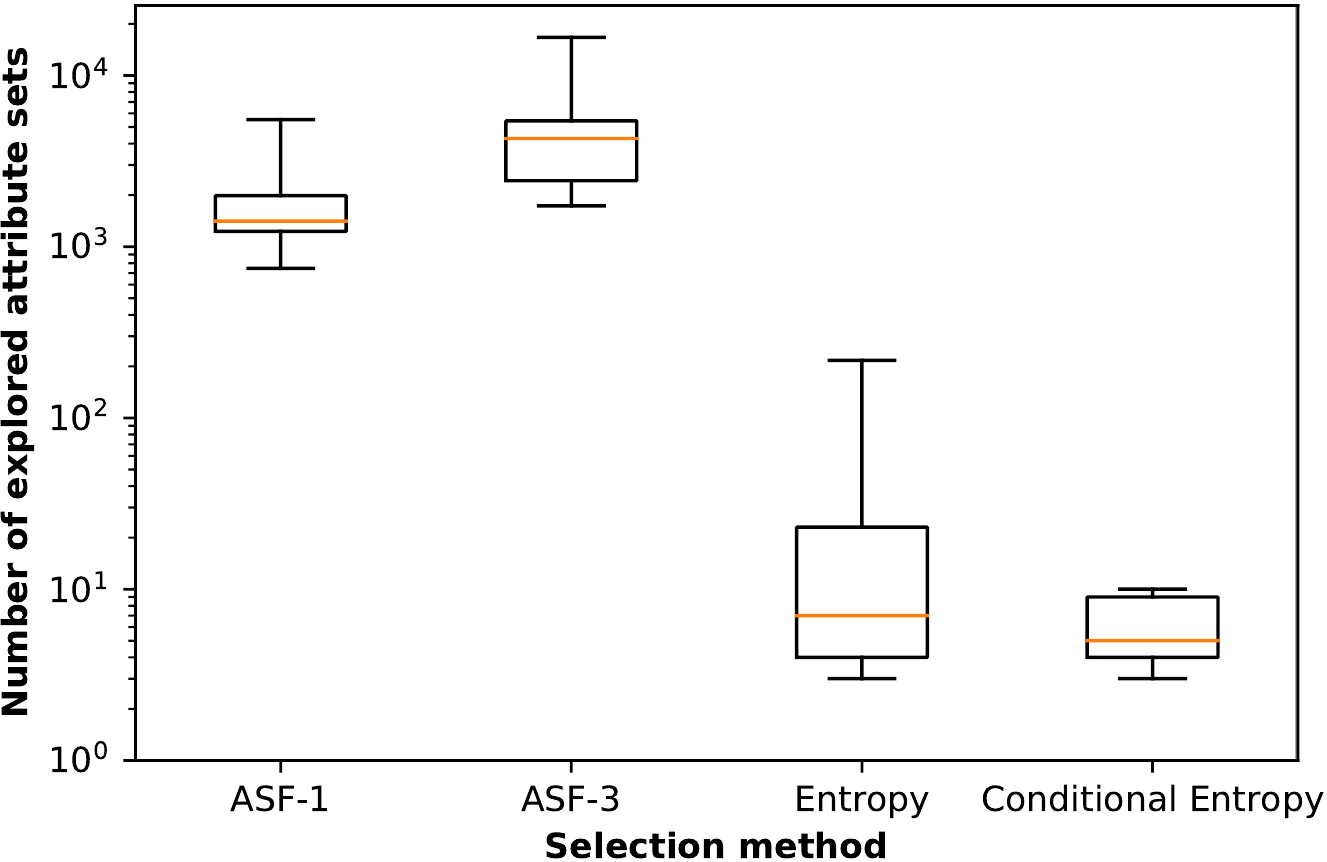}
      \caption{
        Number of explored attribute sets by the attribute selection methods.
      }
      \label{fig:number-of-explored-attribute-sets}
      \Description[
        Number of explored attribute sets by the attribute selection methods.
      ]{
        Number of explored attribute sets by the attribute selection methods.
      }
    \endminipage
    \hfill
    \minipage{\columnwidth}
      \centering
      \includegraphics[width=\columnwidth]{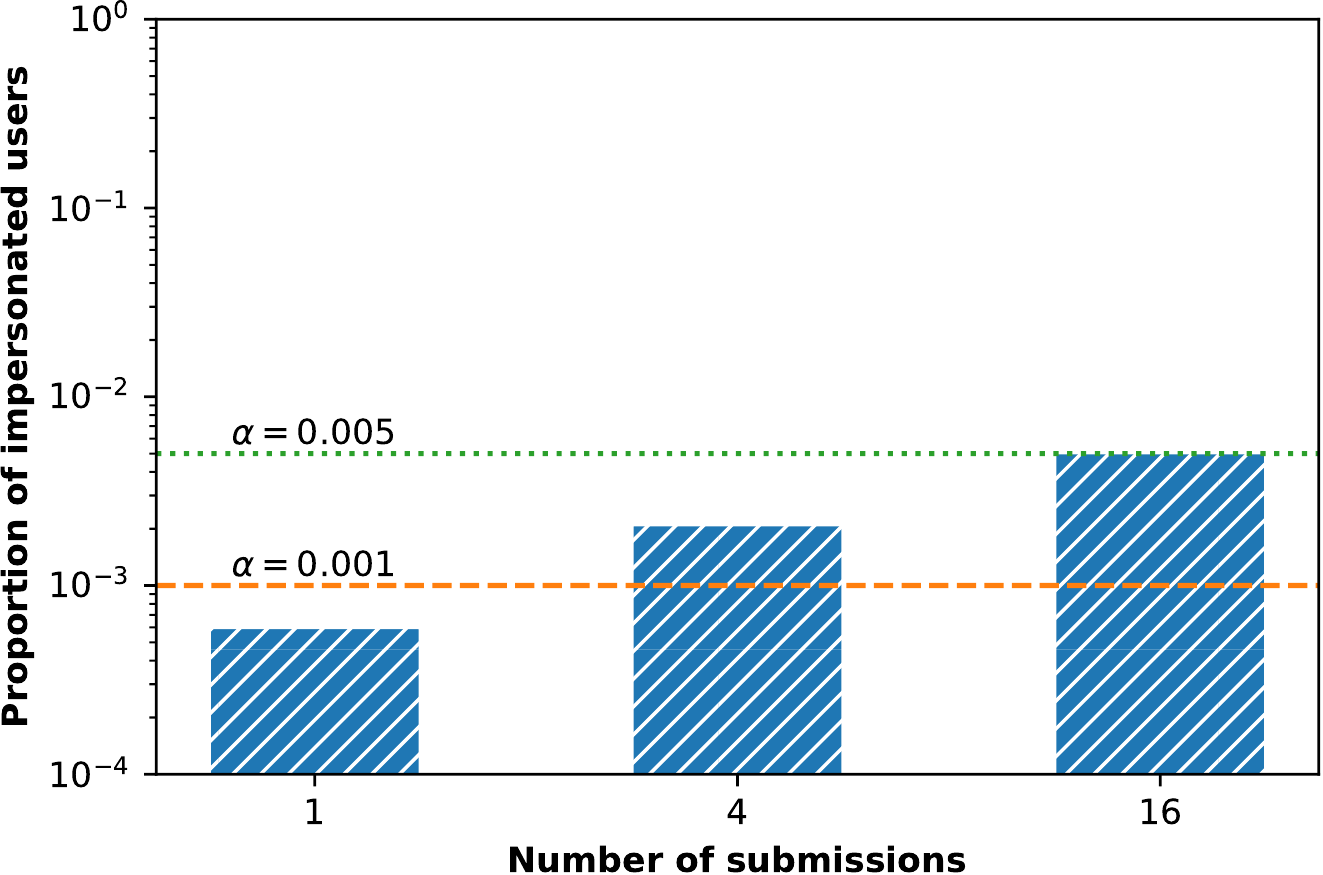}
      \caption{
        Proportion of impersonated users among the $30,000$ users, considering the candidate attributes, the matching function, and as a function of the number of submissions.
        The sensitivity thresholds $\alpha$ that have no solution for some cases are displayed.
      }
      \label{fig:minimum-sensitivity}
      \Description[
        Proportion of impersonated users considering the candidate attributes, the matching function, and as a function of the number of submissions.
      ]{
        Proportion of impersonated users considering the candidate attributes, the matching function, and as a function of the number of submissions.
      }
    \endminipage
  \end{figure*}

  \subsection{Attribute Selection Framework Results}
  \label{sec:attribute-selection-framework-results}
    In this section, we present the results obtained on the previously presented dataset by processing the Attribute Selection Framework on the instantiated attackers, and compare them with the results of the baselines.
    The results are obtained for the $12$ cases consisting of the Cartesian product between the values of the sensitivity threshold $\alpha$ and those of the number of submissions $\beta$.
    We present here the obtained results and discuss the attributes that are the most selected by the framework.
    The exhaustive list of the selected attributes is provided in Appendix~\ref{app:selected-attributes}.

    \subsubsection{Key Results}
    \label{sec:key-results}
      The attribute sets found by the Attribute Selection Framework (ASF) generate fingerprints that are $12$ to $1,663$ times smaller, $9$ to $32,330$ times faster to collect, and with $4$ to $30$ times less changing attributes between two observations, compared to the candidate attributes and on average.
      Compared to the attribute sets found by the baselines, the ones found by the ASF-1 generate fingerprints that are up to $97$ times smaller, are collected up to $3,361$ times faster, and with up to $7.2$ times less changing attributes between two observations, on average.
      These gains come with a higher computation cost, as the ASF-1 explores more attribute sets by three orders of magnitude compared to the baselines.
      However, the implemented attribute sets can be updated rarely, and the usability gain is reflected on each authentication performed by each user.

      Increasing the number of explored paths by the ASF to three does not significantly change the results.
      The attribute sets found by the ASF-3 can have a lower usability cost, or a higher usability cost due to local optimum (see Section~\ref{sec:reasons-for-suboptimal-results}).
      We show that even when considering all of our candidate attributes, the strongest attacker that is able to submit $4$ fingerprints can impersonate $63$ users out of the $30,000$ users of our sample.
      If this attacker is able to submit $16$ fingerprints, this number increases to $152$ users.

    \subsubsection{Results of the Attribute Selection Framework}
      Figure~\ref{fig:asf-cost-breakdown} displays the cost of the attribute sets found by the ASF with $1$ explored path (ASF-1), the ASF with $3$ explored paths (ASF-3), the entropy, and the conditional entropy.
      The costs are in points, so that $10,000$ additional points increase the size of fingerprints by $10$~kilobytes, their collection time by $1$~second, or the number of changing attributes between observations by $1$~attribute, on average.
      There is a solution for $9$ out of the $12$ cases.
      The cases without a solution are discussed in Section~\ref{sec:lower-bound-impersonated-users}.

      Half of the attribute sets found using the ASF-1 generate fingerprints that, on average, have a size lower than $34$~bytes (less than $522$~bytes for all sets), are collected in less than $0.59$ms (less than $1.01$~seconds for all sets), and have less than $0.02$ changing attributes between two observations (less than $0.07$ attributes for all sets).
      Compared to the candidate attributes and on average, the generated fingerprints are $12$ to $1,663$ times smaller, $9$ to $32,330$ times faster to collect, and with $4$ to $30$ times less changing attributes between two observations.

      The difference in usability cost of the attribute sets found by the ASF-3 and the ASF-1 is negligible.
      The attribute sets found by the ASF-3 are as less costly as they are more costly than the attribute sets found by the ASF-1.
      This results in the median additional cost of each dimension being zero.
      The average resulting fingerprint is from $198$~bytes smaller to $30$~bytes larger, takes from $3$ms less to collect to $0.3$ms more, and has from $0.03$ less changing attributes to $0.04$ more.
      Exploring more paths can counter-intuitively provide a higher usability cost, due to the local optimum problem described in Section~\ref{sec:reasons-for-suboptimal-results}.
      Indeed, when exploring more nodes, the followed paths can diverge as we hold more temporary solutions, which can be local optimum.
      The computation cost of increasing the number of explored paths is not worth the expected gain in our experimental setup.

    \subsubsection{Comparison with the Baselines}
      The ASF-1 finds attribute sets that consume less resources than the baselines in all the $9$ cases having a solution.
      The attribute sets found by the \emph{entropy} consume more resources than the attribute sets found by the ASF-1, with a total cost from $1.8$ to $14$ times higher.
      The average generated fingerprint by the attribute sets chosen by the entropy, compared to the attribute sets chosen by the ASF-1, is from $1.6$ to $97$ times larger, has a collection time that is from $1.5$ to $1,872$ times higher, and has from $1.5$ to $7.2$ times more changing attributes between the consecutive fingerprints.
      The attribute sets found by the \emph{conditional entropy} consume more resources than the attribute sets found by the ASF-1, with a total cost from $1.3$ to $15$ times higher.
      The average generated fingerprint by the attribute sets chosen by the conditional entropy, compared to the attribute sets chosen by the ASF-1, is from $1.5$ to $16$ times larger, has a collection time that is from $1.3$ to $3,361$ times higher, and has from $1.1$ to $4.7$ times more changing attributes between the consecutive fingerprints.

    \subsubsection{Reasons for Sub-optimal Results}
    \label{sec:reasons-for-suboptimal-results}
      The sub-optimal solutions that are found by the Attribute Selection Framework are due to a problem of local optimum.
      At a given step, the most efficient attribute sets~$S$ can have supersets of higher cost than the supersets of another~$S'$.
      The supersets of~$S$ are then explored, whereas the less costly supersets here would have been the supersets of~$S'$.

    \subsubsection{Computation Cost}
      The attribute selection framework has a higher computation cost.
      Indeed, at each stage of the exploration, the ASF explores up to ${n-1}$ attribute sets\footnote{
        The set of the explored attribute sets can overlap as two temporary solutions can have a common superset.
      } for each temporary solution, with $n$ being the number of candidate attributes.
      The ASF-1 explores more attribute sets by three orders of magnitude compared to the baselines.
      However, this is an upper bound as the baselines require preprocessing.
      Indeed, the attributes has to be sorted by their entropy or by their conditional entropy.

      Figure~\ref{fig:number-of-explored-attribute-sets} displays the number of attribute sets explored by the attribute selection methods.
      The number of explored attribute sets by the ASF-1 goes from $748$ to $5,522$.
      The ASF-3 explores approximately $3$ times more attribute sets than the ASF-1: from $1,730$ to $16,659$ explored attribute sets.
      The number of attribute sets explored by the entropy ranges from $3$ to $217$, and from $3$ to $10$ for the conditional entropy.
      However, the conditional entropy method requires to sort the $n$ attributes by their conditional entropy, which requires ${\sum_{i=0}^n {n-i}}$ steps.
      This difference of explored attribute sets between the entropy and the conditional entropy is explained by the latter avoiding selecting correlated attributes.

    \subsubsection{Lower Bound on the Impersonated Users}
    \label{sec:lower-bound-impersonated-users}
      The obtained sensitivity against our instantiated attackers ranges from the minimum sensitivity when considering the candidate attributes, as displayed in Figure~\ref{fig:minimum-sensitivity}, to the maximum sensitivity of~$1.0$ when considering no attribute at all.
      All the possible attribute sets have their sensitivity comprised between these two extremum.
      The instantiated attackers that are allowed $4$ submissions are able to impersonate $63$ users out of the $30,000$ users, which exceeds the sensitivity threshold of $0.001$.
      When allowed $16$ submissions, the number of impersonated users increases to $152$, which exceeds the sensitivity threshold of $0.005$ that corresponds to $150$ users.

    \subsubsection{Selected Attributes}
      We have $9$ combinations of sensitivity threshold and number of submissions that show a solution.
      The attribute selection framework is executed twice with two number of explored paths ($1$ and $3$), hence we have $18$ cases for which it found a solution.
      We discuss below the six most selected attributes that are selected in more than five cases.
      We remark that they concern hardware and software components that we expect to not be correlated.
      Indeed, we do not expect a strong link to exist between the browser window size, the number of logical processor cores, the graphics driver, the browser version, a scheme drawn in the browser, and the type of network connection.
      The results about the attributes that we present here come from their analysis by Andriamilanto et al.~\cite{AALG21}.

      The \texttt{innerHeight} property of the \texttt{window} Java\-Script object provides the height of the visible part of the browser window.
      It is selected in all the cases, and is the first attribute to be selected during the exploration as it provides the highest efficiency.
      Indeed, its usability cost is low with, on average, a size of $3.02$~bytes, a collection time of $0.14$ms, and a change between $9.38$\% of the observations.
      Moreover, it is the fifth most distinctive attribute of our dataset, with an entropy of $8.53$~bits and the most common value being shared by $2.70$\% of the fingerprints.

      The \texttt{hardwareConcurrency} property that is collected from the \texttt{navigator} Java\-Script object provides the number of logical processor cores of the device that runs the browser.
      This attribute is selected in $11$ cases, and shows a high efficiency mostly due to the very low usability cost.
      Indeed, it shows, on average, a size of $1$~byte (the value is majoritarily a single digit), a collection time of $0.17$ms, and a change between $0.11$\% of the observations.
      However, it shows a lower distinctiveness, with an entropy of $1.88$~bits and the most common value being shared by $39.64$\% of the fingerprints.

      The \texttt{UNMASKED\_RENDERER\_WEBGL} property of an initialized WebGL Context provides a textual description of the graphics driver.
      This attribute is selected in $8$ cases and shows a high efficiency.
      Indeed, it has, on average, a size of $24.51$~bytes, a collection time of $0.27$ms, and a change between $0.91$\% of the observations.
      Although the most common value is shared by $28.27$\% of the fingerprints, it still provides an entropy of $5.89$~bits.

      The \texttt{appVersion} property of the \texttt{navigator} Java\-Script object provides the version of the browser.
      This attribute is selected in $7$ cases.
      It shows, on average, a size of $101.76$~bytes, a collection time of $0.13$ms, and a change between $1.57$\% of the observations.
      Although the most common value is shared by $22.61$\% of the fingerprints, it still provides an entropy of $7.52$~bits.

      The HTML5 canvas inspired by the AmIUnique study~\cite{LRB16} is selected in $7$ cases, mainly due to its high distinctiveness.
      It has, on average, a size of $63.98$~bytes, a collection time of $71.17$ms, and a change between $1.36$\% of the observations.
      It shows an entropy of $7.76$~bits, and the most common value is shared by $7.09$\% of the fingerprints.

      The \texttt{connection.type} property of the \texttt{navigator} Java\-Script object provides the type of the network connection in use by the browser.
      This attribute is selected in $6$ cases, mainly due to its low usability cost.
      It provides, on average, a size of $1.47$~bytes, a collection time of $0.21$ms, and a change between $0.80$\% of the observations.
      It shows a lower distinctiveness compared to the other attributes, with an entropy of $0.61$~bits and the most common value being shared by $89.52$\% of the fingerprints.

\section{Discussion}
\label{sec:discussion}

  \subsection{Usage of Multiple Browsers}
  \label{sec:usage-of-multiple-browsers}
    Users tend to browse websites using multiple devices, typically a desktop and a mobile device\footnote{
      \url{https://www.javelinstrategy.com/coverage-area/how-online-vs-mobile-shifting-browser-vs-app}
    }.
    FPSelect can be extended to support the usage of multiple browsers by the users, by changing the sensitivity measure so that a user is impersonated if one of his browsers is spoofed by the attacker.
    Using a monotonic matching function (e.g., the matching function described in Section~\ref{sec:implementation-of-the-matching-function}), this sensitivity measure is also monotonic\footnote{
      Indeed, if the fingerprint of one of the user's browsers matches with the fingerprint of the attacker for an attribute set~$C$, it also matches for any subset of~$C$.
    }.
    Boda et al.~\cite{BFGI12} showed that some attributes provide information about the underlying system (e.g., the list of fonts) and can be used for cross-browser fingerprinting.
    Although such technique is interesting in an authentication context (e.g., recognizing the common attributes between the browsers of a user), this is out of the scope of our work.

  \subsection{Update of Attributes through Time}
  \label{sec:update-attributes-through-time}
    The verifier can keep the attributes of the fingerprinting probe up to date by re-executing the framework.
    To do so, she performs a fingerprint collection on a browser population close to the population of her web platform, using a wide-surface of fingerprinting attributes.
    We emphasize that the usability requirement is less strict for such experiment (e.g, the fingerprints can take more time or more space).
    Web technologies do not evolve frequently.
    For example, Andriamilanto et al.~\cite{AALG21} analyze the dataset that we study, and do not observe any significant change over the $6$ months of the experiment.
    Moreover, changing the attribute set requires to update the fingerprint that is stored for each user.
    Hence, the verifier can -- and should -- perform this process rarely (e.g., once per semester or per year).
    Finally, the verifier can monitor the distinctiveness and the stability of the stored fingerprints, and perform an update if a drastic change is detected (e.g., an attribute becomes highly unstable~\cite{KGBRF16} or homogeneous).

  \subsection{Attribute Sets in a Per-Browser Basis}
  \label{sec:per-browser-attribute-sets}
    The attribute set can also be selected in a per-browser basis, but FPSelect is not designed for this.
    However, it is possible to execute FPSelect on subpopulations of browsers (e.g., mobile and desktop browsers) to obtain an attribute set per subpopulation.
    To do so, the whole framework is simply executed on the subpopulation of browsers.
    The sensitivity measure then considers that the attacker focus on this subpopulation (i.e., he knows the fingerprint distribution of this subpopulation).
    The costs are also specifically measured on the subpopulation (e.g., the list of plugins is most of the time empty for the mobile browsers~\cite{SPJ15, AALG21}, hence is less costly for this subpopulation).
    The thresholds of the subsets have to be set so that the overall sensitivity threshold is satisfied.
    The simplest way is to set the thresholds of the subsets to the overall threshold.
    Indeed, if less than $x$ percent of the users of each subset are impersonated, less than $x$ percent of the overall users are.

\section{Related Works}
\label{sec:related-works}

  \subsection{Attribute Selection}
    Previous works identify the need to reduce the included attributes, and used various methods to perform the selection.
    Most of the previous works either remove the attributes of the lowest entropy~\cite{VLRR18}, or leverage greedy algorithms that iteratively pick the attributes of the highest weight (typically the entropy) until a threshold (typically on the number of attributes) is reached~\cite{MEN11, KZW15, BRC16, HRA18, THS18}.
    These methods do not consider the correlation that can occur between the attributes.
    Indeed, two attributes can separately have a high entropy, but, when taken together, provide a lower entropy than another attribute set.
    Table~\ref{tab:fingerprint-dataset-example} displays a concrete example of such case.

    To the best of our knowledge, two works take the correlation into account in their attribute selection method.
    Fifield and Egelman~\cite{FE15} weigh the attributes by the conditional entropy given the attributes that are already picked.
    The conditional entropy of each attribute is updated on each turn, and picking two correlated attributes is therefore avoided.
    Pugliese et al.~\cite{PRGB20} propose two attribute selection methods that iteratively pick the attribute that maximizes a criterion, given the attributes that are already chosen.
    The first criterion to maximize is the number of users for which their fingerprints are not shared by any other user and stay identical between at least two observations.
    The second criterion to maximize is the duration for which these fingerprints stay identical.
    These two works only maximize one criterion, and ignore the usability cost of the attributes.
    On the contrary, our framework performs a trade-off between the sensitivity that is tied to the distinctiveness, and the usability cost that has a stability dimension.

    Gulyás et al.~\cite{GAC16} study the problem of finding the set of $s$-items (e.g., fonts, plugins, applications) for which to check the presence on a device, to reduce the number of devices that agree on the same value.
    They prove that this problem is NP-hard, and propose greedy algorithms to find the closest approximation in polynomial time.
    Our problem is different because we do not choose the items to check the presence for, that consist of binary value, but on selecting the categorical attributes to collect.
    Moreover, they seek to reduce the number of collected attributes, and to minimize the users that agree on the same values.
    We seek to reduce the sensitivity against dictionary attackers, and to minimize the usability cost that comprises various aspects.
    Indeed, the attributes are not equal regarding their usability cost (e.g., some are collected almost instantly whereas others take seconds).

    Flood and Karlsson~\cite{FK12} evaluate every possible attribute set obtained from their $13$ attributes to find the set that provides the best classification results.
    An exhaustive search is feasible on a small set of candidate attributes, but unrealistic on a larger set.
    Indeed, there are $2^n$ possible attribute sets for $n$~candidate attributes.

  \subsection{Evaluation of the Sensitivity of an Attribute Set}
    Alaca et al.~\cite{AV16} rate fingerprinting attributes given properties that include the resource usage and the resistance to spoofing.
    They also model attackers according to different strategies and knowledge.
    The rating of the attributes mainly comes from estimations, and most of their spoof-resistant attributes are outside our boundaries on attribute choice.
    Indeed, we only consider the attributes collected via HTTP headers or JavaScript properties, that are accessible without permission and do not directly concern the user (e.g., IP address, geolocation) but rather his web browsing platform.

    Laperdrix et al.~\cite{LABN19} propose a challenge-response mechanism based on dynamic attributes which values also depend on provided instructions (e.g., the HTML5 canvas depends on drawing instructions).
    They identify various attacks that can be executed on an authentication mechanism that includes browser fingerprinting.
    The attacks notably include the submission of the most common fingerprints.
    We also consider the attacker that submits the most common fingerprints following his knowledge, and consider in addition a matching function to measure his reach in a realistic context (i.e., a fingerprint can spoof others that are similar).

\section{Conclusion}
\label{sec:conclusion}
  In this study, we propose FPSelect, a framework for a verifier to tailor his fingerprinting probe by picking the attribute set that limits the sensitivity against an instantiated attacker, and reduces the usability cost.
  We formalize the Attribute Selection Problem that the verifier has to solve, show that it is a generalization of the Knapsack Problem, model the potential solutions as a lattice of attribute sets, and propose a greedy exploration algorithm to find a solution.
  We evaluate FPSelect on a real-life browser fingerprint dataset, and compare it with common attribute selection methods that rely on the entropy and the conditional entropy.
  The attribute sets found by FPSelect generate fingerprints that are $12$ to $1,663$ times smaller, $9$ to $32,330$ times faster to collect, and with $4$ to $30$ times less changing attributes between two observations, compared to the candidate attributes and on average.
  Compared to the baselines, the attribute sets found by FPSelect generate fingerprints that are up to $97$ times smaller, are collected up to $3,361$ times faster, and with up to $7.2$ times less changing attributes between two observations, on average.

  In future works, we will first extend our attack model with the attackers that possess targeted knowledge about users (e.g., the value of fixed attributes, components of their web environment).
  Indeed, the attacker that manages to infer the fingerprint distribution of small subpopulations (e.g., grouped by the operating system), and to link users to a subpopulation, would obtain a more skewed distribution which could help him to extend his reach.
  Moreover, the attacker that has the knowledge of previous challenges of dynamic attributes and the associated responses, can try to forge the response to an unseen challenge using other methods (e.g., image processing for the canvas).
  The study of the ability of these attackers, and the measure of the sensitivity against them, are let as future works.
  Second, the behavior of our framework on other experimental setups (browser population, dataset, measures, parameters) would be interesting, and is let as future works.

\begin{acks}
  We want to thank the anonymous reviewers for their usefull reviews; Benoît Baudry, David Gross-Amblard, Joris Duguépéroux, and Louis Béziaud for their valuable comments; and Alexandre Garel for his work on the experiment.
\end{acks}

\bibliographystyle{ACM-Reference-Format}
\bibliography{main}

\appendix

\section{Demonstrations of Measures Monotonicity}
\label{app:sensitivity-usability-cost-monotonicity-demonstrations}

  \subsection{Monotonicity of the Illustrative Sensitivity}
    \begin{theorem}
      \label{thm:sensitivity-monotonicity}
      Considering the same limit on the number of submissions~$\beta$, the mapping from users to their browser fingerprint~$\mathcal{M}$, the probability mass function~$p$, and the set of matching functions~$\Phi$.
      For any couple of attribute sets $C_k$ and $C_l$ so that ${C_k \subset C_l}$, we have ${\mathrm{s}(C_k) \geq \mathrm{s}(C_l)}$ when measuring the sensitivity using Algorithm~\ref{alg:sensitivity-measure}.
    \end{theorem}

    \begin{proofsketch}
    \label{proof:sensitivity-monotonicity}
      We focus on a fingerprint~${f \in F_{C_l}}$ when considering the attribute set~$C_l$, and the dictionary~$V$ used to attack~$f$.
      We project~$f$ to the attribute set~$C_k$ to obtain the fingerprint~${g \in F_{C_k}}$.
      This fingerprint~$g$ can be attacked using the dictionary~$W$ composed of the fingerprints of~$V$ projected to~$C_k$.
      As the matching function works in an attribute basis, if a fingerprint~${h \in V}$ matches with~$f$, we have ${f[a] \approx^a h[a] : \forall a \in C_l}$ that is true.
      As $C_k$~is a subset of~$C_l$, we also have ${f[a] \approx^a h[a] : \forall a \in C_k}$ that is true.
      The projection of $h$ to $C_k$ then spoofs $g$, hence the fingerprints spoofed when considering~$C_l$ are also spoofed when considering~$C_k$.
      The sensitivity when considering~$C_k$ is therefore at least equal to that when considering~$C_l$.

      When projecting the fingerprints of the attack dictionary to the attribute set~$C_k$, some of them can come to the same fingerprint.
      In which case, more fingerprints can be added to the dictionary until reaching the submission limit~$\beta$.
      This can lead to more spoofed fingerprints and impersonated users.
      The sensitivity when considering~$C_k$ can be higher than when considering~$C_l$.

      For any attribute sets $C_k$ and $C_l$ so that ${C_k \subset C_l}$, we then have ${\mathrm{s}(C_k) \geq \mathrm{s}(C_l)}$ when measuring the sensitivity using Algorithm~\ref{alg:sensitivity-measure}.
    \end{proofsketch}

  \subsection{Monotonicity of the Illustrative Usability Cost}
    \begin{theorem}
    \label{thm:cost-monotonicity}
      For a given fingerprint dataset~$D$, and attribute sets $C_k$ and $C_l$ so that ${C_k \subset C_l}$, we have ${\mathrm{cost}(C_k, D) < \mathrm{cost}(C_l, D)}$.
    \end{theorem}

    \begin{proof}
    \label{proof:complete-cost-monotonicity}
      We consider $C_i$ and $C_j$ two attribute sets, so that they differ by the attribute~$a_d$ such that ${C_j = C_i \cup \{ a_d \}}$.
      We consider the fingerprint dataset~$D$ from which the measures are obtained.

      The memory cost of~$C_j$ is strictly greater than the memory cost of~$C_i$.
      As $C_j$ and $C_i$ differ by the attribute~$a_d$, we have
      \begin{equation}
        \mathrm{mem}(C_j, D) =
          \mathrm{mem}(C_i, D) + \frac{
            \sum_{(b, f) \in D} \mathrm{size}(f[a_d])
          }{
            \mathrm{card}(D)
          }
      \end{equation}
      The size of the attribute~$a_d$ is strictly positive, hence we have the inequality ${\mathrm{mem}(C_j, D) > \mathrm{mem}(C_i, D)}$.

      The temporal cost of~$C_j$ is greater than or equal to the temporal cost of~$C_i$, as adding an attribute cannot reduce the collection time.
      The attribute $a_d$~can be sequential or asynchronous, and can take identical or more time than the current longest attribute of~$C_i$.
      Below, we explore all these cases.
      (1)~If $a_d$~is asynchronous, we have the following cases:
      (a)~$a_d$~takes less time than the longest asynchronous attribute~$a_l$, then the collection time is either that of~$a_l$ or the total of the sequential attributes, so $a_d$~does not influence the collection time and ${\mathrm{time}(C_j, D) = \mathrm{time}(C_i, D)}$,
      (b)~$a_d$~takes more time than the longest asynchronous attribute, but less or equal to the total of the sequential attributes, then the maximum is the total of the sequential attributes and ${\mathrm{time}(C_j, D) = \mathrm{time}(C_i, D)}$,
      (c)~$a_d$~takes more time than both the longest asynchronous attribute and the total of the sequential attributes, then the collection time is that of~$a_d$, and ${\mathrm{time}(C_j, D) > \mathrm{time}(C_i, D)}$.
      (2)~If $a_d$~is sequential, we have the following cases:
      (a)~$a_d$~increases the total collection time of the sequential attributes, but the total stays below that of the longest asynchronous attribute, then we have the equality ${\mathrm{time}(C_j, D) = \mathrm{time}(C_i, D)}$,
      (b)~$a_d$~increases the total collection time of the sequential attributes, which is then higher than that of the longest asynchronous attribute, then ${\mathrm{time}(C_j, D) > \mathrm{time}(C_i, D)}$\footnote{
        Either the total collection time of the sequential attributes of $C_i$ does not exceed that of its longest asynchronous attribute, and adding $a_d$ results in the total collection time of the sequential attributes surpassing that of the longest asynchronous attribute.
        In this case, ${\mathrm{time}(C_j, D) > \mathrm{time}(C_i, D)}$.
        Either the total collection time of the sequential attributes of $C_i$ exceeds that of its longest asynchronous attribute.
        As $a_d$ increases the total collection time of the sequential attributes, we have ${\mathrm{time}(C_j, D) > \mathrm{time}(C_i, D)}$.
      }.
      These are all the possible cases, hence ${\mathrm{time}(C_j, D) \geq \mathrm{time}(C_i, D)}$.

      The instability cost of~$C_j$ is greater than or equal to that of~$C_i$, as either $a_d$~is completely stable and ${\mathrm{ins}(C_j, D) = \mathrm{ins}(C_i, D)}$, otherwise $a_d$~is unstable and ${\mathrm{ins}(C_j, D) > \mathrm{ins}(C_i, D)}$.
      We then have ${\mathrm{ins}(C_j, D) \geq \mathrm{ins}(C_i, D)}$.

      As the cost weight vector~$\gamma$ is composed of strictly positive numbers, the cost of~$C_j$ is therefore strictly higher than the cost of~$C_i$ due to the memory cost.
      Recursively, it holds for any~$C_k$ and~$C_l$ so that ${C_k \subset C_l}$, hence the cost is monotonic.
      For any fingerprint dataset~$D$, and attribute sets $C_k$ and $C_l$ so that ${C_k \subset C_l}$, we have ${\mathrm{cost}(C_k, D) < \mathrm{cost}(C_l, D)}$.
    \end{proof}

\section{Usability Cost of Attributes}
\label{app:attributes-usability-cost}
  In this appendix, we discuss the distribution of the three usability cost dimensions among the attributes.
  For more insight into the fingerprints and the attributes of the dataset of this study, we refer to the study of Andriamilanto et al.~\cite{AALG21} which includes an analysis of this dataset.

  \begin{figure}
    \centering
    \includegraphics[width=0.8\columnwidth]{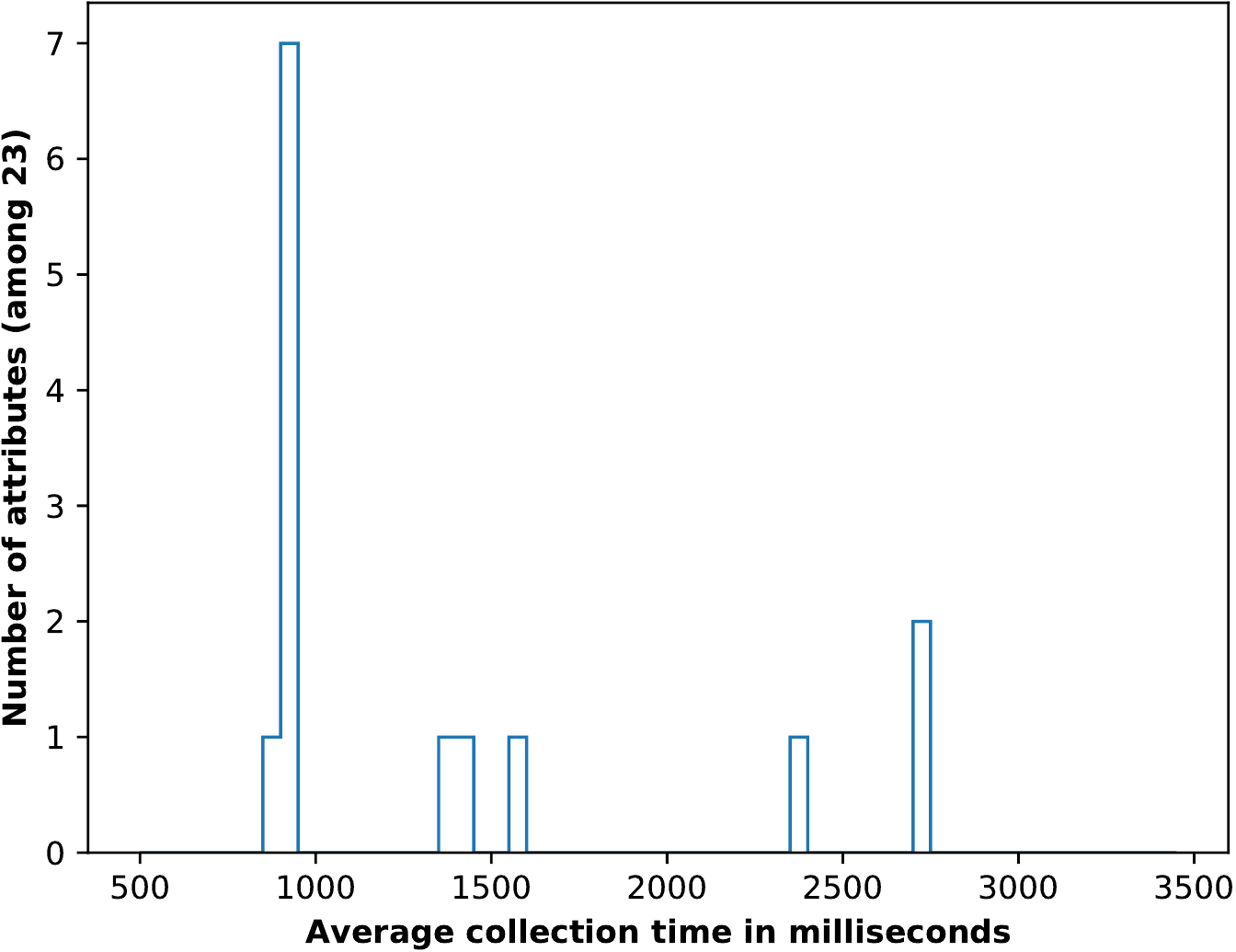}
    \caption{Distribution of the average collection time among the $23$~asynchronous attributes.}
    \label{fig:attributes-asynchronous-collection-time}
    \Description[
      Distribution of the average collection time among the $23$~asynchronous attributes.
    ]{
      Distribution of the average collection time among the $23$~asynchronous attributes.
    }
  \end{figure}

  \begin{figure*}
    \minipage{0.32\textwidth}
      \includegraphics[width=\linewidth]{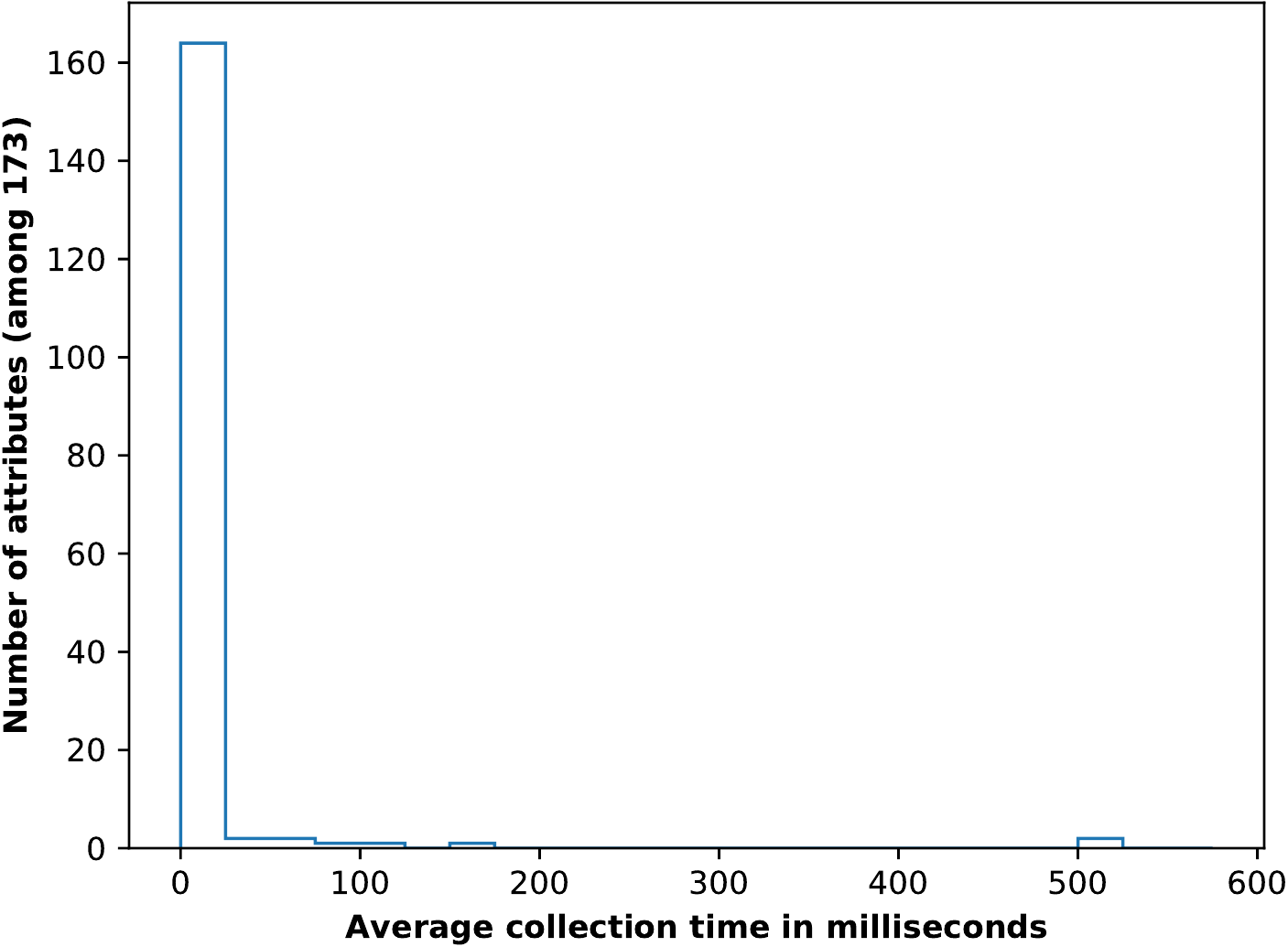}
      \caption{Distribution of the average collection time among the $173$~sequential attributes.}
      \label{fig:attributes-sequential-collection-time}
      \Description[
        Distribution of the average collection time among the $173$~sequential attributes.
      ]{
        Distribution of the average collection time among the $173$~sequential attributes.
      }
    \endminipage
    \hfill
    \minipage{0.32\textwidth}
      \includegraphics[width=\linewidth]{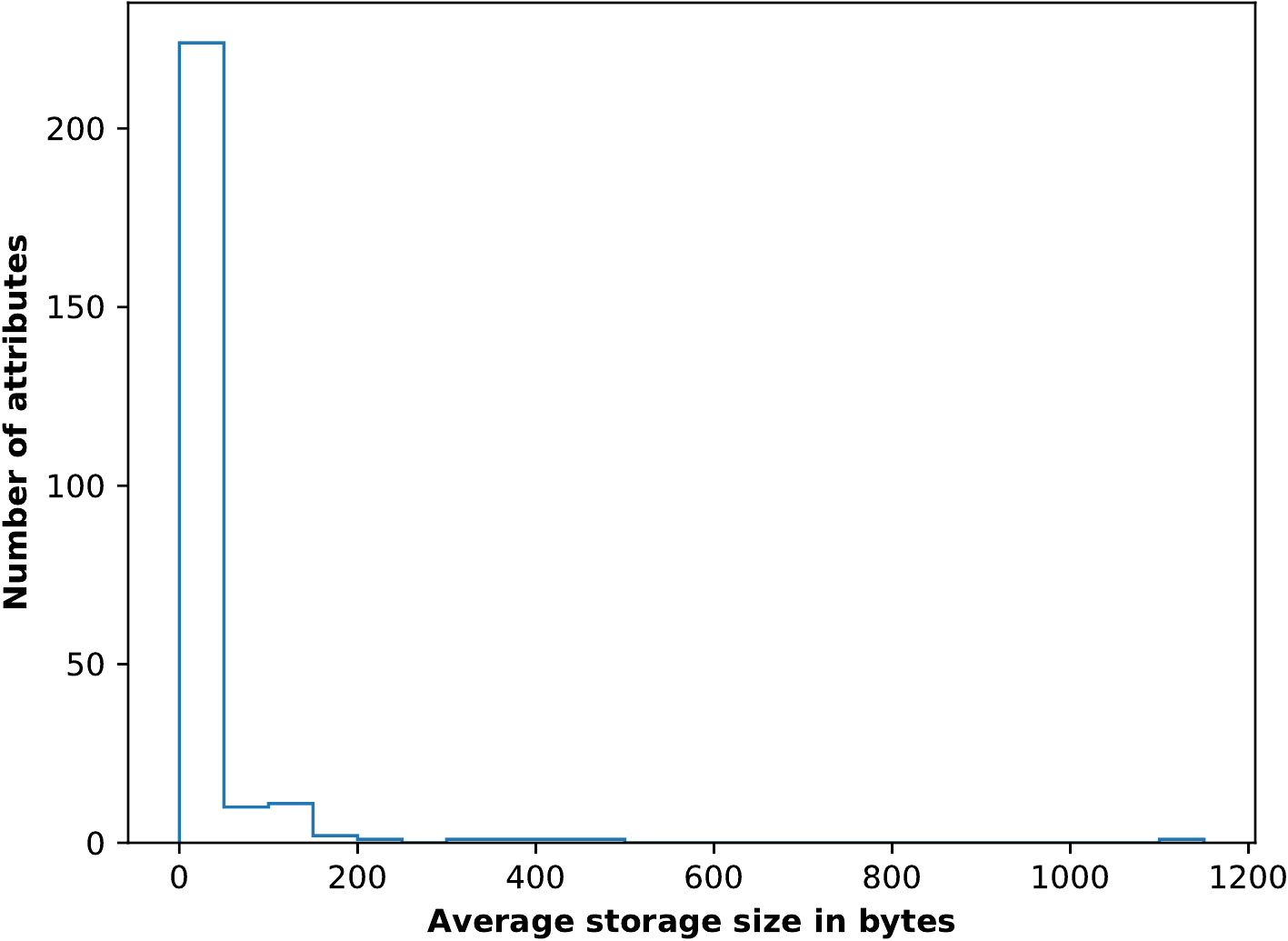}
      \caption{Distribution of the average storage size among the attributes.}
      \label{fig:attributes-storage-size}
      \Description[
        Distribution of the average storage size among the attributes.
      ]{
        Distribution of the average storage size among the attributes.
      }
    \endminipage
    \hfill
    \minipage{0.32\textwidth}
      \includegraphics[width=\linewidth]{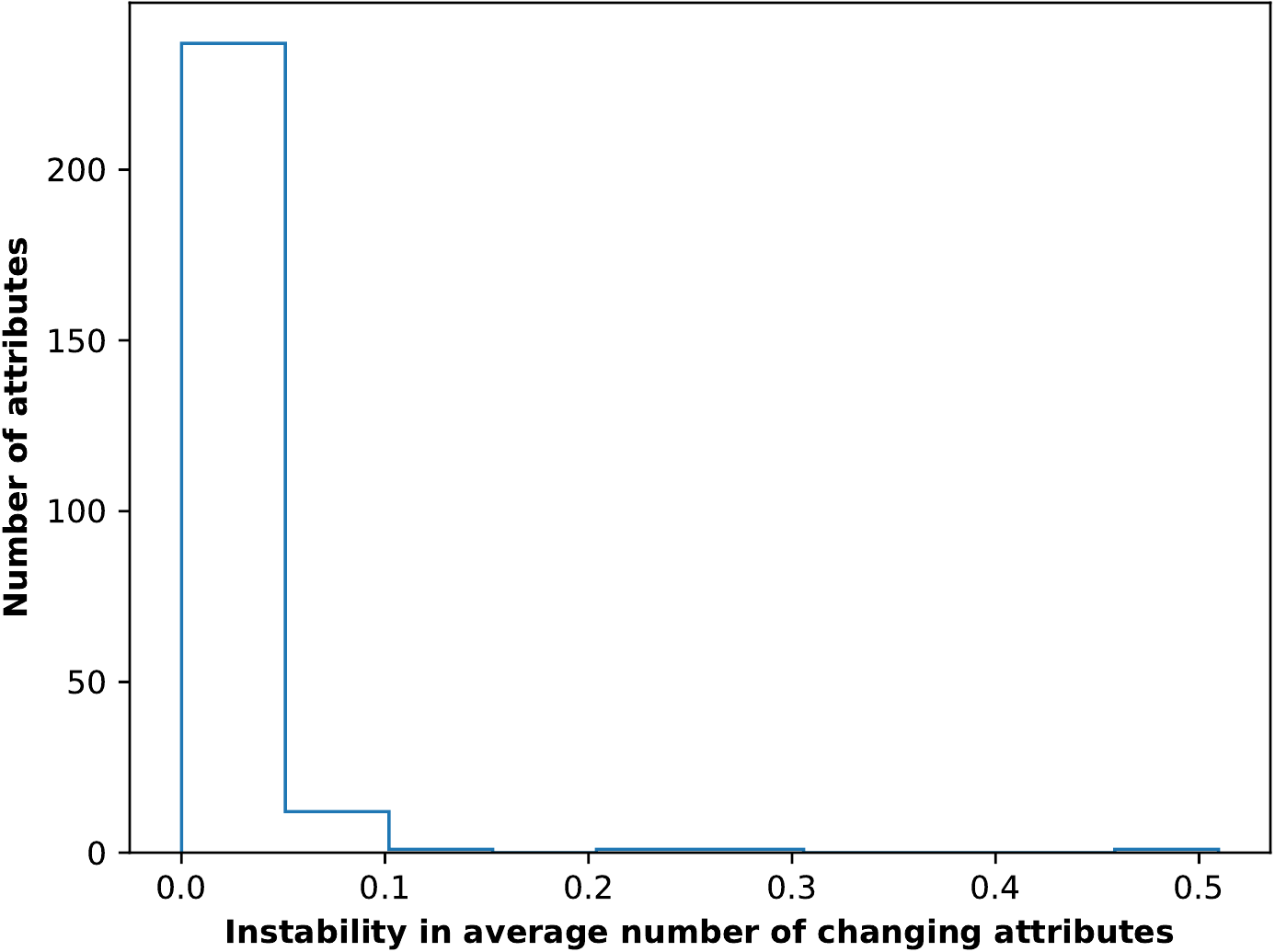}
      \caption{Distribution of the average instability among the attributes.}
      \label{fig:attributes-instability}
      \Description[
        Distribution of the average instability among the attributes.
      ]{
        Distribution of the average instability among the attributes.
      }
    \endminipage
  \end{figure*}

  Figure~\ref{fig:attributes-asynchronous-collection-time} presents the distribution of the average collection time among the $23$~asynchronous attributes.
  These attributes comprise the extension detection methods that require to wait for the web page to render~\cite{SVS17, SN17}, heavy processes like the WebRTC fingerprinting method~\cite{FE16, ALF17} that creates dummy connections, and the audio fingerprinting methods~\cite{QF19}.
  Figure~\ref{fig:attributes-sequential-collection-time} presents the distribution of the average collection time among the $173$~sequential attributes.
  Only $9$~attributes take more than $25$~milliseconds on average to collect, among which we retrieve the six canvases~\cite{MS12, LRB16}, the list of WebGL extensions, the list of available streaming codecs, and the string representation of a specific date.
  The sequential and asynchronous attributes do not sum to the $253$~candidate attributes.
  We do not show the HTTP headers that are collected instantly, which have a null collection time.
  Moreover, we only display the collection time of the attributes source of the extracted attributes, as they have the same collection time.
  Indeed, the source attributes are the one that are actually collected, and the extracted attributes are inferred from them.

  Figure~\ref{fig:attributes-storage-size} presents the distribution of the average storage size among the $253$~attributes.
  We have $29$~attributes that weigh more than $50$~bytes on average.
  They are the list attributes (e.g., list of plugins), the verbose properties (e.g., the \texttt{UserAgent}), and the string representation of the hashed canvases.

  Figure~\ref{fig:attributes-instability} presents the distribution of the average instability of the $253$~attributes.
  Only $16$~attributes have more than $0.05$~changes per observation on average.
  They comprise the attributes related to the screen resolution, the size of the browser window, the canvases, the experimental WebRTC fingerprinting method, the list of speech synthesis voices, the \texttt{Cache-Control} HTTP header, and the attribute that stores the HTTP headers that are not stored in a dedicated attribute.

\section{Attributes Selected by the Attribute Selection Framework}
\label{app:selected-attributes}
  In this appendix, we provide the list of the $21$ attributes that are selected by the attribute selection framework.
  Table~\ref{tab:attributes-selected-by-the-asf} lists the selected attributes.
  We name them according to the same nomenclature as~\cite{AALG21}.
  We denote \emph{N} the \texttt{navigator} object, \emph{S} the \texttt{screen} object, \emph{W} the \texttt{window} object, and \emph{A} an initialized \emph{Audio Context}.
  Finally, we denote \emph{WG} an initialized \emph{WebGL Context}, and \emph{WM} the \texttt{WG.MAX\_} prefix.
  To get the WG.[...].UNMASKED\_RENDERER\_WEBGL attribute, we first get an identifier named \texttt{id} from the unmasked property of the \texttt{get\-Extension('WEB\-GL\-\_\-de\-bug\-\_\-ren\-de\-rer\-\_\-in\-fo')} object, and then get the actual value by calling \texttt{get\-Parameter(id)}.
  We use square brackets so that \texttt{A.[B, C]} means that the property is accessed through \texttt{A.B} or \texttt{A.C}.

  \begin{table*}
    \centering
    \begin{tabular}{lcccccccccccccccccc}
      \toprule
        \multirow{6}{*}{Selected Attribute} & \multicolumn{9}{c}{$k=1$} & \multicolumn{9}{c}{$k=3$} \\
      \cmidrule(lr){2-10} \cmidrule(lr){11-19}
        & \multicolumn{4}{c}{$\beta=1$} & \multicolumn{3}{c}{$\beta=4$} & \multicolumn{2}{c}{$\beta=16$} & \multicolumn{4}{c}{$\beta=1$} & \multicolumn{3}{c}{$\beta=4$}  & \multicolumn{2}{c}{$\beta=16$} \\
      \cmidrule(lr){2-5} \cmidrule(lr){6-8} \cmidrule(lr){9-10}
      \cmidrule(lr){11-14} \cmidrule(lr){15-17} \cmidrule(lr){18-19}
        & \rotatebox{90}{$\alpha=0.001$} & \rotatebox{90}{$\alpha=0.005$} & \rotatebox{90}{$\alpha=0.015$} & \rotatebox{90}{$\alpha=0.025$} & \rotatebox{90}{$\alpha=0.005$} & \rotatebox{90}{$\alpha=0.015$} & \rotatebox{90}{$\alpha=0.025$} & \rotatebox{90}{$\alpha=0.015$} & \rotatebox{90}{$\alpha=0.025$} & \rotatebox{90}{$\alpha=0.001$} & \rotatebox{90}{$\alpha=0.005$} & \rotatebox{90}{$\alpha=0.015$} & \rotatebox{90}{$\alpha=0.025$} & \rotatebox{90}{$\alpha=0.005$} & \rotatebox{90}{$\alpha=0.015$} & \rotatebox{90}{$\alpha=0.025$} & \rotatebox{90}{$\alpha=0.015$} & \rotatebox{90}{$\alpha=0.025$} \\
      \midrule
        W.innerHeight & $\bullet$ & $\bullet$ & $\bullet$ & $\bullet$ & $\bullet$ & $\bullet$ & $\bullet$ & $\bullet$ & $\bullet$ & $\bullet$ & $\bullet$ & $\bullet$ & $\bullet$ & $\bullet$ & $\bullet$ & $\bullet$ & $\bullet$ & $\bullet$ \\
        N.hardwareConcurrency & $\bullet$ & & $\bullet$ & & $\bullet$ & $\bullet$ & & $\bullet$ & $\bullet$ & & & $\bullet$ & & $\bullet$ & $\bullet$ & $\bullet$ & $\bullet$ & \\
        WG.[...].UNMASKED\_RENDERER\_WEBGL & & $\bullet$ & & & & $\bullet$ & $\bullet$ & $\bullet$ & & & $\bullet$ & & & & $\bullet$ & & $\bullet$ & $\bullet$ \\
        N.appVersion & $\bullet$ & & & & $\bullet$ & & & $\bullet$ & & $\bullet$ & & & & $\bullet$ & & & $\bullet$ & $\bullet$ \\
        HTML5 canvas inspired by AmIUnique (PNG) & $\bullet$ & & & & $\bullet$ & & & $\bullet$ & $\bullet$ & $\bullet$ & & & & & & & $\bullet$ & $\bullet$ \\
        N.connection.type & $\bullet$ & & & & $\bullet$ & & & $\bullet$ & & $\bullet$ & & & & $\bullet$ & & & $\bullet$ & \\
        N.plugins & $\bullet$ & & & & & & & $\bullet$ & $\bullet$ & $\bullet$ & & & & & & & $\bullet$ & \\
        {[[N, W].doNotTrack, N.msDoNotTrack]} & & $\bullet$ & & $\bullet$ & $\bullet$ & & & & & & $\bullet$ & & & & & & & \\
        Accept-Encoding HTTP header & $\bullet$ & & & & & & & & & $\bullet$ & & & & $\bullet$ & & & & \\
        Height of first bounding box & $\bullet$ & & & & & & & & & $\bullet$ & & & & & & & & \\
        Origin of a created \texttt{div} & & & & & $\bullet$ & & & & & & & & & $\bullet$ & & & & \\
        W.ontouchstart support & & & & & & & & & & & & & $\bullet$ & & & $\bullet$ & & \\
        W.[performance, console].jsHeapSizeLimit & & & & & $\bullet$ & & & & & & & & & & & & & \\
        S.width & & & & & & & & & & $\bullet$ & & & & & & & & \\
        W.openDatabase support & & & & & & & & & & & & & $\bullet$ & & & & & \\
        WM.COMBINED\_TEXTURE\_IMAGE\_UNITS & & & & & & & & & & & & & & $\bullet$ & & & & \\
        N.platform & & & & & & & & & & & & & & $\bullet$ & & & & \\
        HTML5 canvas similar to Morellian (PNG) & & & & & & & & & & & & & & $\bullet$ & & & & \\
        Accept-Language HTTP header & & & & & & & & & & & & & & $\bullet$ & & & & \\
        WM.CUBE\_MAP\_TEXTURE\_SIZE & & & & & & & & & & & & & & & & $\bullet$ & & \\
        A.sampleRate & & & & & & & & & & & & & & & & $\bullet$ & & \\
      \bottomrule
    \end{tabular}

    \caption{
      The attributes selected by the attribute selection framework for each experimentation setup.
      We denote $k$ the number of explored paths, $\beta$ the number of fingerprint submissions, and $\alpha$ the sensitivity threshold.
    }
    \label{tab:attributes-selected-by-the-asf}
  \end{table*}

\end{document}